%% file: main.tex
\documentclass[aps,pra,11pt,superscriptaddress,tightenlines,longbibliography,nofootinbib]{revtex4-2}
\usepackage[utf8]{inputenc}
\usepackage[T1]{fontenc}
\usepackage[colorlinks = true]{hyperref}
\usepackage{xcolor}
\definecolor{darkred}  {rgb}{0.5,0,0}
\definecolor{darkblue} {rgb}{0,0,0.5}
\definecolor{darkgreen}{rgb}{0,0.5,0}
\definecolor{myblue}   {rgb}{0.176,0.643,0.937}
\definecolor{myred}    {rgb}{0.557,0.176,0.188}
\hypersetup{
  urlcolor   = blue,         
  linkcolor  = darkblue,     
  citecolor  = darkgreen,    
  filecolor   = darkred       
}
\usepackage[caption=false]{subfig}
\usepackage{url}            
\usepackage{booktabs}       
\usepackage{amsfonts}       
\usepackage{nicefrac}       
\usepackage{microtype}      
\usepackage{qcircuit}
\newcommand{\cgate}[1]{*+<.6em>{#1} \POS ="i","i"+UR;"i"+UL **\dir{-};"i"+DL **\dir{-};"i"+DR **\dir{-};"i"+UR **\dir{-},"i" \cw} 
\usepackage{amsmath,amssymb}
\usepackage{amsthm}
\usepackage[capitalise,nameinlink]{cleveref}
\usepackage{tikz}
\usetikzlibrary{arrows.meta,positioning}
\usepackage{verbatim}
\usepackage{physics}
\usepackage{multirow}
\usepackage{makecell}
\usepackage{mathtools}
\usepackage{nicematrix}
\usepackage{graphics}
\usepackage{placeins}
\usepackage{circledsteps}

\newtheorem{definition}{Definition}
\newtheorem{theorem}{Theorem}
\newtheorem{lemma}{Lemma}
\newtheorem{corollary}{Corollary}
\newtheorem{proposition}{Proposition}

\definecolor{ginger}{rgb}{0.69, 0.4, 0.0}

\newcommand{\lket}[1]{\vert #1 \rangle\!\rangle}
\newcommand{\lbra}[1]{\langle\!\langle #1 \vert}

\newcommand{\mc}{\mathcal}

\begin{document}

\title{A generalized cycle benchmarking algorithm for characterizing\\ mid-circuit measurements}
\author{Zhihan Zhang}
\email{zhihan-z21@mails.tsinghua.edu.cn}
\affiliation{Institute for Interdisciplinary Information Sciences, Tsinghua University, Beijing 100084, China}

\author{Senrui Chen}
\email{csenrui@uchicago.edu}
\affiliation{Pritzker School of Molecular Engineering, The University of Chicago, Chicago, Illinois 60637, USA}

\author{Yunchao Liu}
\email{yunchaoliu@berkeley.edu}
\affiliation{Department of Electrical Engineering and Computer Sciences, University of California, Berkeley, California 94720, USA}

\author{Liang Jiang}
\email{liang.jiang@uchicago.edu}
\affiliation{Pritzker School of Molecular Engineering, The University of Chicago, Chicago, Illinois 60637, USA}

\begin{abstract}
Mid-circuit measurements (MCMs) are crucial ingredients in the development of fault-tolerant quantum computation. While there have been rapid experimental progresses in realizing MCMs, a systematic method for characterizing noisy MCMs is still under exploration. In this work we develop a cycle benchmarking (CB)-type algorithm to characterize noisy MCMs. The key idea is to use a joint Fourier transform on the classical and quantum registers and then estimate parameters in the Fourier space, analogous to Pauli fidelities used in CB-type algorithms for characterizing the Pauli noise channel of Clifford gates. Furthermore, we develop a theory of the noise learnability of MCMs, which determines what information can be learned about the noise model (in the presence of state preparation and terminating measurement (SPAM) noise) and what cannot, which shows that all learnable information can be learned using our algorithm. As an application, we show how to use the learned information to test the independence between measurement noise and state preparation noise in an MCM. Finally, we conduct numerical simulations to illustrate the practical applicability of the algorithm. Similar to other CB-type algorithms, we expect the algorithm to provide a useful toolkit that is of experimental interest.
\end{abstract}

\maketitle

\section{Introduction}
Mid-circuit measurements (MCMs) are a central component in quantum error correction and are being developed on various hardware platforms~\cite{ryan2021realization,singh2023mid,google2023suppressing,bluvstein2024logical}. To perform syndrome measurements, one needs to initialize fresh ancilla qubits, apply the syndrome extraction circuit, and measure the ancilla qubits while maintaining the coherence of system qubits. Characterizing the noise model for this process is thus a necessary step for fault-tolerant quantum computation. 
For quantum gate characterization, a standard approach is to engineer the noise process into Pauli channels~\cite{wallman2016noise,hashim2020randomized} and various methods have been developed to characterize the Pauli channel~\cite{flammia2020efficient,erhard2019characterizing,emerson2007symmetrized,flammia2022averaged,van2023probabilistic,carignan2023error}. However, MCMs are much more complicated than a noisy gate: the measurement readout can be faulty, the post-measurement state can be faulty, and there can be additional errors on unmeasured qubits. Despite significant recent efforts~\cite{govia2023randomized,rudinger2022characterizing}, developing general and systematic characterization methods for this more complicated noise model of MCMs still remains open.

In this paper we address this challenge by developing a general algorithm for characterizing MCMs. Our algorithm is a cycle benchmarking (CB)-type protocol~\cite{erhard2019characterizing,flammia2020efficient,emerson2007symmetrized,flammia2022averaged} generalizing CB with interleaved gates~\cite{chen2023learnability,van2023probabilistic,carignan2023error}, the state-of-the-art method for characterizing Clifford gates. The key idea behind CB with interleaved gates is the observation that the Fourier transform of the Pauli error rates (known as Pauli fidelities) are eigenvalues of the Pauli channel and therefore are easy to learn via repeated application of the noisy gate. We generalize this view to MCMs by defining Pauli fidelities as a joint Fourier transform on the classical and quantum registers, which provides a means of converting between the physical space (focusing on Pauli errors) and dual space (focusing on Fourier coefficients). A key challenge here is that the Pauli fidelities are no longer eigenvalues of the noise channel and are not directly learnable. We address this by introducing an additional Fourier transform during post-processing to filter out undesired components. This allows us to use the CB framework and repeatedly apply MCMs to learn products of Pauli fidelities. See \cref{fig:circuitExample} for a comparison between our algorithm and CB with interleaved gates.

On the other hand, not all parameters of noisy MCMs are learnable in a way that is robust to state preparation and terminating measurement (SPAM\footnote{Note that, while SPAM typically refers to ``state preparation and measurements'', in our context, we treat MCMs as a mid-circuit gadget different from the usual destructive measurement. Thus, the ``M'' in SPAM should be interpreted specifically as the terminating measurements.}) noise, similar to noisy Clifford gates~\cite{chen2023learnability}, due to the existence of gauge freedom~\cite{nielsen2021gate}. We give an exact and complete characterization about the learnable parameters of MCMs by generalizing the graph-theoretical framework of Ref.~\cite{chen2023learnability}, and show that our algorithm is able to learn all learnable parameters of noisy MCMs.

As a concrete application, our algorithm gives a general method to test whether a noisy MCM can be factorized as a measure-and-prepare instrument, which can help understand the underlying physical mechanism for the implementation of MCMs. See \cref{sec:summary} for a more detailed summary.

\subsection{General noisy MCMs}
We start by defining the noise model associated with MCMs. We consider an ideal MCM as a general Clifford gate $G$ on $n+m$ qubits, followed by a measurement on the $n$ ancilla qubits in the computational basis (see \cref{fig:MCMIllu}). The other $m$ qubits are left unmeasured. Note that the quantum register for the $n$ ancilla qubits is still present after the measurement, i.e. we keep the post-measurement quantum state. 
We consider this general model because MCMs are often used for qubit reset~\cite{bluvstein2024logical,baumer2023efficient,hashim2023quasi,mcewen2021removing}.
Also note that this setting incorporates subsystem measurements ($G$ equals to identity), syndrome measurements ($G$ is the syndrome extraction circuit), as well as Clifford gates without measurement ($n=0$) as special cases.

In the literature of noise characterization, noisy Clifford gates are often modeled as an ideal Clifford gate followed by a Pauli channel (channels where random Pauli error occur) or equivalently, a Pauli channel followed by an ideal Clifford gate.
This is due to the randomized compiling techniques~\cite{wallman2016noise,hashim2020randomized} which can twirl a general CPTP noise channel into a Pauli channel that is more tractable for characterization.
Motivated by recent developments in randomized compiling for subsystem measurements~\cite{beale2023randomized}, we extend these techniques to MCMs.
We show that general noisy MCMs can be twirled into the form of an ideal Clifford gate $G$ followed by a uniform stochastic instrument $\mathcal{U}$, which is a highly structured noisy measurement and its physical meaning can be understood as the following three parts:
\begin{enumerate}
    \item Suppose we obtain the result $k$ during measurement readout. From this we know that ideally, the $n$ ancilla qubits should be projected into the basis state $\ket{k}$ (throughout the paper, $\ket{k}$ means the computational basis state of index $k$). However, due to noise in the readout, the ancilla qubits are actually projected into a different basis state $\ket{k+a}$, where $a$ is an $n$-bit string that came from some probability distribution. 
    \item After the ancilla qubits are projected into $\ket{k+a}$, the noise can further affect the post-measurement state, which becomes $\ket{k+b}$, where $b$ is an $n$-bit string that came from some probability distribution.
    \item Finally, the measurement process causes an additional random Pauli error $P$ acting on the remaining $m$ qubits due to crosstalk.
\end{enumerate}
The above noise model is parameterized by Pauli error rates, which is a joint probability distribution $p_{a,b}^P$ satisfying $\sum_{P\in\mathcal{P}^m,a,b\in\mathbb{Z}_2^n}p_{a,b}^P=1$ ($\mathcal{P}^m$ is the $m$-qubit Pauli group).

Our goal is to learn the Pauli error rates, which has $4^{n+m}-1$ unknown parameters.
More specifically, on input density matrix $\rho$ acting on $n+m$ qubits, the output density matrix equals
\begin{equation}\label{eq:noisemodel}
\rho\mapsto \sum_{P\in\mathcal{P}^m,k,a,b\in\mathbb{Z}_2^n} p_{a,b}^P (P\otimes \ketbra{k+b}{k+a})G\rho G^\dag (P\otimes \ketbra{k+a}{k+b})\otimes\ketbra{k}_R.
\end{equation}
Here the third register $R$ is a classical register for measurement readout. A key point here is that the Pauli error rates $p_{a,b}^P$ are \emph{independent} from the readout result $k$; this is a consequence of randomized compiling.

\begin{figure}[t]
    \subfloat[Ideal and noisy Clifford gates\label{fig:CliffordIllu}]{\centering
\begin{minipage}[t][1.6cm][c]{0.4\textwidth}
\begin{tabular}{c}
\Qcircuit @C=1em @R=.7em {
    & /^n \qw & \gate{G} & \qw & \xrightarrow[version]{noisy} & & /^n \qw & \gate{G} & \gate{\Lambda} & \qw\\
    & & & & \push{\rule{3em}{0em}} & & & & &
}
\end{tabular}
\end{minipage}
}\subfloat[Ideal and noisy MCMs\label{fig:MCMIllu}]{\centering
\begin{minipage}[t][1.6cm][c]{0.5\textwidth}
\begin{tabular}{c}
\Qcircuit @C=1em @R=.7em {
    & /^{m} \qw & \multigate{1}{G} & \qw & \qw & & \push{\rule{3em}{0em}} & & /^{m} \qw & \multigate{1}{G} & \multimeasure{1}{\mathcal{U}} & \qw\\
    & /^{n} \qw & \ghost{G} & \meter & \qw & & \xrightarrow[version]{noisy} & & /^{n} \qw & \ghost{G} & \ghost{\mathcal{T}} & \qw\\
    & & & \control \cwx & \cw & k_0 & \push{\rule{3em}{0em}} & & & & \control \cwx & \cw & k\\
}
\end{tabular}
\end{minipage}
}\\
\subfloat[Measure-and-prepare instruments\label{fig:MaPIllu}]{\centering
\Qcircuit @C=1em @R=.7em {
    & /^m \qw & \multigate{1}{G} & \qw & \qw & \qw & \qw & \qw & & \push{\rule{3em}{0em}} & & /^m \qw & \multigate{1}{G} & \multimeasure{1}{\mathcal{U}} & \qw & \qw & \multigate{1}{\mathcal{E}} & \qw & \qw & \\
    & /^n \qw & \ghost{G} & \meter & \measureD{discard} & \lstick{\ket{0^n}} & \gate{X} & \qw & & \xrightarrow[version]{noisy} & & /^n \qw & \ghost{G} & \ghost{\mathcal{U}} & \measureD{discard} & \lstick{\ket{0^n}} & \ghost{\mathcal{E}} & \gate{X} & \qw & \\
    & & & \control \cwx & \cw & \cw & \control \cwx \cw & \cw & k_0 & \push{\rule{3em}{0em}} & & & & \control \cwx & \cw & \cw & \cw & \control \cwx \cw & \cw & k \\
    & & & & \push{\rule{7.5em}{0em}} & & & & & & & & & & \push{\rule{7.5em}{0em}} & & & &
}
}
\caption{Noise models considered in this work. (a) A noisy Clifford gate $G$ is modeled as an ideal gate $G$ followed by a Pauli channel $\Lambda$. (b) The noisy version of an MCM, whose corresponding ideal version consists of a Clifford gate $G$ followed by a subsystem measurement, is modeled as the ideal Clifford gate $G$ followed by a uniform stochastic instrument $\mathcal{U}$. The measurement outcome in the noiseless case is denoted $k_0$ whereas the outcome in the noisy case is denoted $k$. (c) We also consider a special family of noisy MCMs called measure-and-prepare instruments, in which the noise model can be decomposed into measurement noise and state-preparation noise.}
\label{fig:noisemodels}
\end{figure}

\subsection{Measure-and-prepare instruments}\label{sec:USIRIntro}
A natural subclass of the general noise model in \cref{eq:noisemodel} is the family of noise channels that can be factorized into measurement noise and state-preparation noise, which we call measure-and-prepare instruments (\cref{fig:MaPIllu}). 
For example, in certain experimental setups, MCMs are realized in a measure-and-prepare manner (e.g., destructive measurements with replaced free ancilla in atom array platforms~\cite{bluvstein2024logical}).
After an MCM, the ancillas are initialized in ground state $\ket{0}$, which are then adjusted into the state corresponding to the measurement result through the application of $X$ gates (NOT gates).
This sequence of operations results in a special type of MCM and gives us additional structures among the parameters.

A fundamental question about the physical process that implements general MCMs is whether it has the above factorized form. As an application of our result, in~\Cref{sec:application} we will give a general method for testing whether an unknown MCM can be represented by a measure-and-prepare instrument.

\subsection{Summary of results}
\label{sec:summary}
In the remainder of the paper, we will first introduce notations and preliminaries in~\cref{sec:notation}.
Subsequent sections will detail specific aspects of our results.

\paragraph{Fourier transforms and dual space formalism} (\cref{sec:Fourier}).
For noisy Clifford gates, standard characterization methods such as some CB-type protocols~\cite{chen2023learnability,van2023probabilistic,carignan2023error,flammia2020efficient} employ a Fourier transform, converting Pauli errors rates into quantities known as Pauli fidelities, which are eigenvalues of the Pauli channel and are easy to learn. Our first technical contribution is to generalize this conversion between the physical space (Pauli errors) and the dual (Fourier) space, by defining a suitable notion of ``Pauli fidelities'' for MCMs, which is a joint Fourier transform on the classical and quantum registers.\\

\paragraph{Generalized CB algorithm for characterizing MCMs} (\cref{sec:characterization}). Similar to CB with interleaved gates, our algorithm learns the Pauli fidelities and reconstructs the Pauli error rates via an inverse Fourier transform. However, the Pauli fidelities of MCMs are no longer eigenvalues of the noise channel, so learning these Pauli fidelities requires new techniques. To address this challenge, we further introduce a key technical idea: an additional Fourier transform on the classical registers during the data processing procedure, which allows us to learn a single desired component from sums of many components. This completes a generalized version of the CB algorithm to characterize MCMs.\\

\paragraph{Learnability using pattern transfer graph} (\cref{sec:learnability}). 
Furthermore, we extend the pattern transfer graph technique from~\cite{chen2023learnability} to determine which parameters in the noise model are SPAM-robustly learnable and which are not.
We show that our algorithm captures all the information that can be learned, which corresponds to the cycle space of the pattern transfer graph.
Conversely, information that our algorithm cannot learn is inherently unlearnable by any SPAM-robust algorithm.
Corresponding numerical simulation results are shown in \cref{fig:learnedCycle}.\\

\paragraph{Testing measure-and-prepare channels} (\cref{sec:application}).
As an application of our algorithm, we show that it can be used to determine whether an arbitrary MCM can be factorized as a measure-and-prepare instrument. This shows that our algorithm can be used to probe the fundamental physical process of the measurement process. Numerical results are shown in \cref{fig:learnedIndep}.\\

\paragraph{Numerical results} (\cref{sec:numerical}).
Finally, we conduct numerical simulations for the algorithm and applications to support our analysis.

\section{Notations and preliminaries}\label{sec:notation}

\textbf{Pauli operators, Clifford gates.}
In this paper, $\mathcal{P}^m:=\{I,X,Y,Z\}^m$ represents the $m$ qubit Pauli group (modulus phase), $\mathcal{C}$ represents the group of single qubit Clifford gates.

\textbf{Pauli operator with vector exponents.}
For $P\in\mathcal{P}^1$ and $x\in\mathbb{Z}_2^n$, we use $P^x$ to represent the $n$ qubit Pauli operator that act as $I$ on the $i$th qubit if the $i$th bit of $x$ is $0$, and act as $P$ on the $i$th qubit if the $i$th bit of $x$ is $1$.

\textbf{Inner product of Pauli operators.}
For $P,Q\in\mathcal{P}^m$, we define inner product $\langle P,Q\rangle$ to be $0$ if Pauli operators $P$ and $Q$ commute, and $1$ otherwise.\footnote{This is the symplectic inner product between the symplectic representation of $P$ and $Q$. Specifically, the symplectic representation $\phi$ of a Pauli operator $P\in\mathcal{P}^m$ is the unique length $2m$ bit string that satisfies $P=i^{\phi(P)^T\Upsilon\phi(P)}X^{\phi(P)_{\mathrm{odd}}}\cdot Z^{\phi(P)_{\mathrm{even}}}$, where $\phi(P)_{\mathrm{odd}}$ and $\phi(P)_{\mathrm{even}}$ denotes the substring of $\phi(P)$ consists of odd and even positions, repectively and $\Upsilon=\bigoplus_{j=1}^m\left(\begin{smallmatrix}0&1\\0&0\end{smallmatrix}\right)$. The inner product is defined as the binary symplectic form $\langle P,Q\rangle=\phi(P)^T(\Upsilon+\Upsilon^T)\phi(Q)\mod2$.}
That is, $PQ=(-1)^{\langle P,Q\rangle}QP$.

\textbf{Pauli weight patterns.}
Following Ref.~\cite{chen2023learnability}, for $P\in\mathcal{P}^m$, we define its Pauli weight pattern $pt(P)$ to be the $m$ bit string that is $0$ on the $i$th bit if $P$ act as identity on the $i$th qubit, and is $1$ on the $i$th bit if $P$ act as non-identity on the $i$th qubit.

\textbf{Quantum channels.}
Calligraphic letters represent channels, for example $\mathcal{E},\mathcal{M},\mathcal{T},\mathcal{U}$.
Specifically, $\mathcal{G}$, $\mathcal{H}$, $\mathcal{P}$, $\mathcal{X}$, $\mathcal{Z}$ represent specific channels of applying gates of the corresponding normal letters $G$, $H$, $P$, $X$, $Z$.
For example, $\mathcal{G}(\rho)=G\rho G^\dagger$. 

\textbf{Pauli channels.}
Pauli channels are represented by $\Lambda$ exclusively.
They are channels of applying a random Pauli operator according to a probability distribution,
\begin{equation}\label{eq:PCp}
\Lambda=\sum_{P\in\mathcal{P}^m}p^P\mathcal{P},
\end{equation}
where $\{p^P\}_{P\in\mathcal{P}^m}$ are called \textbf{Pauli error rates}.

\textbf{Super operator formalism.}
General treatment of the super operator formalism can be found in e.g.~\cite{greenbaum2015introduction}.
In short, vectorization refers to the isomorphism $|i\rangle\langle j|\mapsto|ij\rangle$. The inner product over the vectorized space is then given by $\langle\!\langle{O|P}\rangle\!\rangle=\Tr(O^\dagger P)$.
In this paper, for $k\in\mathbb{Z}_2^n$, $\lket{k}$ is the vectorization of computational basis state $|k\rangle\langle k|$.
And for $P\in\mathcal{P}^m$, double brackets $|P\rangle\!\rangle$ represents the vectorization of Pauli operator $P$. 


\textbf{Pauli fidelities.}
Recall from \cref{eq:PCp} that a Pauli channel is of the form $\Lambda=\sum_{P\in\mathcal{P}^m}p^P\mathcal{P}$.
Note that the Pauli operators are eigenvectors of the channel:
\begin{equation}
\Lambda|Q\rangle\!\rangle=\sum_{P\in\mathcal{P}^m}p^P\mathcal{P}|Q\rangle\!\rangle=\sum_{P\in\mathcal{P}^m}(-1)^{\langle P,Q\rangle}p^P|Q\rangle\!\rangle.
\end{equation}
The corresponding eigenvalues are defined as Pauli fidelities~\cite{flammia2020efficient}
\begin{equation}\label{eq:PCp2d}
\lambda^Q=\sum_{P\in\mathcal{P}^m}(-1)^{\langle P,Q\rangle}p^P.
\end{equation}
For normalized Pauli channels, the trace preserving condition $\sum_{p\in\mathcal{P}^m}p^P=1$ translates into $\lambda^I=1$.
The inverse transformation is given by
\begin{equation}\label{eq:PCd2p}
p^P=\frac{1}{4^m}\sum_{Q\in\mathcal{P}^m}(-1)^{\langle P,Q\rangle}\lambda^Q.
\end{equation}
Using the Pauli fidelities, we can write the Pauli channel in the dual space
\begin{equation}\label{eq:Pcdng}
\Lambda=\frac{1}{2^m}\sum_{Q\in\mathcal{P}^m}\lambda^Q|Q\rangle\!\rangle\langle\!\langle Q|.
\end{equation}
This uses the Pauli-transfer-matrix representation, which is diagonal for Pauli channels.
We say that the Pauli error rates lie in the physical space because they represent physical errors.
Correspondingly, we say that the Pauli fidelities lie in the dual space.

\textbf{Quantum instruments.}
Mathematically, general MCMs are modeled as quantum instruments, which are sets of completely positive, trace non-increasing maps $\{\mathcal{E}_j\}$ such that $\sum_j\mathcal{E}_j$ is trace preserving.
The action of it on a density operator $\rho$ is~\cite{wilde2011classical}
\begin{equation}
\rho\mapsto\sum_j\mathcal{E}_j(\rho)\otimes|j\rangle\langle j|.
\end{equation}
That is, with probability $p_j=\Tr(\mathcal{E}_j(\rho))$, outcome $j$ is observed and the output state becomes $\frac{1}{p_j}\mathcal{E}_j(\rho)$.

\section{Noise models and their Fourier transform}
\label{sec:Fourier}
\subsection{Uniform stochastic instrument}
As mentioned in the introduction, the motivation for modeling MCMs as uniform stochastic instruments is randomized compiling.
But we will leave the details to~\cref{app:MCMrc} and begin our discussion with a formal definition of the uniform stochastic instruments.

Using the super operator formalism, we can rewrite our noise model in~\cref{eq:noisemodel} to the form of quantum instruments.
\begin{definition}[Uniform stochastic instruments, Eq. 17 of Ref.~\cite{beale2023randomized}]
In physical space, the uniform stochastic instrument is the set $\{\mathcal{U}_k\}$ where each $\mathcal{U}_k$ is expressed by the following trace non-increasing map
\begin{align}
\mathcal{U}_k&\coloneqq\sum_{a,b\in\mathbb{Z}_2^n,P\in\mathcal{P}^m}p_{a,b}^P\mathcal{P}\otimes\lket{k+b}\lbra{k+a}\label{eq:USIp}\\
&=\sum_{a,b\in\mathbb{Z}_2^n}\Lambda_{a,b}\otimes\lket{k+b}\lbra{k+a}\label{eq:USIc}
\end{align}
here $p_{a,b}^P$ is a joint probability distribution on $\mathbb{Z}_2^n\times\mathbb{Z}_2^n\times\mathcal{P}^m$ called Pauli error rates and $\Lambda_{a,b}$ are unnormalized Pauli channels indexed by $a,b$ (note that they are independent of $k$).
\end{definition}

The overall noise model in~\cref{eq:noisemodel} can then be reexpressed as
\begin{equation}
    \lket{\rho}\mapsto \sum_{k\in\mathbb{Z}_2^n}\mc U_k\mc G\lket{\rho}\otimes \lket{k}_R.
\end{equation}
Let $\mc T_k:=\mc U_k\mc G$.
As the Pauli channel $\Lambda_{a,b}$ is independent of $k$, in the following we will often drop the third register $\lket{k}_R$ and work with a specific map $\mc T_k$.
Then we can see that $\{\mathcal{U}_k\}$ and $\{\mathcal{T}_k\}$ are quantum instruments since $\sum_{a,b\in\mathbb{Z}_2^n,P\in\mathcal{P}^m}p_{a,b}^P=1$.

Now we generalize the physical vs dual space picture to uniform stochastic instruments.
We define the Pauli fidelities $\lambda_{x,y}^Q$ as the Fourier transform of the Pauli error rates 
\begin{equation}\label{eq:USIp2d}
\lambda_{x,y}^Q=\sum_{a,b\in\mathbb{Z}_2^n,P\in\mathcal{P}^m}(-1)^{a\cdot x+b\cdot y+\langle P,Q\rangle}p_{a,b}^P.
\end{equation}
We note that the Pauli fidelities defined above are different from the Pauli fidelities of the unnormalized Pauli channels in \cref{eq:USIc} as there is an additional Fourier transformation on the subscript $a$ and $b$.
The role of the Pauli fidelities will become clear after the next lemma.

In fact, our protocol will first learn the Pauli fidelities and then use the inverse transformation
\begin{equation}\label{eq:USId2p}
p_{a,b}^P=\frac{1}{4^{n+m}}\sum_{x,y\in\mathbb{Z}_2^n,Q\in\mathcal{P}^m}(-1)^{a\cdot x+b\cdot y+\langle P,Q\rangle}\lambda_{x,y}^Q
\end{equation}
to obtain the Pauli error rates.

Using Pauli fidelities, we can write the uniform stochastic instrument and thus the instrument $\{\mathcal{T}_k\}$ of our model in the dual space.
\begin{lemma}\label{lem:unif}
Let $\mc U_k$ be a uniform stochastic instrument defined in \cref{eq:USIc}, and let $\mc T_k=\mc U_k \mc G$ where $\mc G$ is a (Clifford) unitary channel. Then $\mathcal{T}_k$ can be expressed in dual space as
\begin{equation}\label{eq:USId}
\mathcal{T}_k=\frac{1}{2^{2n+m}}\sum_{x,y\in\mathbb{Z}_2^n,Q\in\mathcal{P}^m}(-1)^{k\cdot(x+y)}\lambda_{x,y}^Q|Q\otimes Z^y\rangle\!\rangle\langle\!\langle\mathcal{G}^\dagger(Q\otimes Z^x)|.
\end{equation}
\end{lemma}
\begin{proof}
For a uniform stochastic instrument, we have
\begin{align}
\mathcal{U}_k=&\sum_{a,b\in\mathbb{Z}_2^n}\Lambda_{a,b}\otimes\lket{k+b}\lbra{k+a}\\
=&\sum_{a,b\in\mathbb{Z}_2^n}\Lambda_{a,b}\otimes\left(\frac{1}{2^n}\sum_{Q\in\mathcal{P}^n}|Q\rangle\!\rangle\langle\!\langle Q|\right)\lket{k+b}\lbra{k+a}\left(\frac{1}{2^n}\sum_{P\in\mathcal{P}^n}|P\rangle\!\rangle\langle\!\langle P|\right)\label{eq:basischange}\\
=&\sum_{a,b\in\mathbb{Z}_2^n}\Lambda_{a,b}\otimes\left(\frac{1}{2^n}\sum_{y\in\mathbb{Z}_2^n}|Z^y\rangle\!\rangle\langle\!\langle Z^y|\right)\lket{k+b}\lbra{k+a}\left(\frac{1}{2^n}\sum_{x\in\mathbb{Z}_2^n}|Z^x\rangle\!\rangle\langle\!\langle Z^x|\right)\label{eq:simplifyzero}\\
=&\frac{1}{2^{2n}}\sum_{x,y\in\mathbb{Z}_2^n}(-1)^{k\cdot(x+y)}\sum_{a,b\in\mathbb{Z}_2^n}(-1)^{a\cdot x+b\cdot y}\Lambda_{a,b}\otimes|Z^y\rangle\!\rangle\langle\!\langle Z^x|\label{eq:Z-k_inner}\\
=&\frac{1}{2^{2n}}\sum_{x,y\in\mathbb{Z}_2^n}(-1)^{k\cdot(x+y)}\sum_{a,b\in\mathbb{Z}_2^n}(-1)^{a\cdot x+b\cdot y}\left(\frac{1}{2^m}\sum_{P,Q\in\mathcal{P}^m}(-1)^{\langle P,Q\rangle}p_{a,b}^P|Q\rangle\!\rangle\langle\!\langle Q|\right)\otimes|Z^y\rangle\!\rangle\langle\!\langle Z^x|\label{eq:paulirewrite}\\
=&\frac{1}{2^{2n+m}}\sum_{x,y\in\mathbb{Z}_2^n,Q\in\mathcal{P}^m}(-1)^{k\cdot(x+y)}\left(\sum_{a,b\in\mathbb{Z}_2^n,P\in\mathcal{P}^m}(-1)^{a\cdot x+b\cdot y+\langle P,Q\rangle}p_{a,b}^P\right)|Q\otimes Z^y\rangle\!\rangle\langle\!\langle Q\otimes Z^x|\\
=&\frac{1}{2^{2n+m}}\sum_{x,y\in\mathbb{Z}_2^n,Q\in\mathcal{P}^m}(-1)^{k\cdot(x+y)}\lambda_{x,y}^Q|Q\otimes Z^y\rangle\!\rangle\langle\!\langle Q\otimes Z^x|.\label{eq:fidelityform}
\end{align}
\cref{eq:basischange} is a basis change.
\cref{eq:simplifyzero} uses the fact that Pauli operators that have non-zero inner product (over the vectorized space) with a computational basis state must be a tensor product of $I$ and $Z$s.
The phases in \cref{eq:Z-k_inner} comes from the inner products $\langle\!\langle Z^y\mid k+b\rangle\!\rangle$ and $\langle\!\langle k+a\mid Z^x\rangle\!\rangle$.
\cref{eq:paulirewrite} rewrites the Pauli channels using \cref{eq:Pcdng}.

As can be seen from \cref{eq:fidelityform}, the Pauli fidelities are \emph{not} eigenvalues of the uniform stochastic instrument.
They connect the Pauli operators $Q\otimes Z^y$ and $Q\otimes Z^x$.

Now put an ideal gate $G$ before this instrument.
Since $\mathcal{G}^\dagger(Q\otimes Z^x)$ becomes $Q\otimes Z^x$ after applying $\mathcal{G}$, $\langle\!\langle Q\otimes Z^x|\mathcal{G}=\langle\!\langle\mathcal{G}^\dagger(Q\otimes Z^x)|$, we get the result.
\end{proof}
When noise rates are small, we would expect that $p_{0,0}^I$ is close to $1$, while other Pauli error rates are small.
This means that all Pauli fidelities will be close to $1$.
Note that the trace preserving condition $\sum_{P\in\mathcal{P}^m,a,b\in\mathbb{Z}_2^n}p_{a,b}^P=1$ translates into $\lambda_{0,0}^I=1$ here.

\begin{table}[htp]
\centering
\resizebox{0.85\textwidth}{!}{
\begin{NiceTabular}{|c|c|c|c|}
\hline
& & Physical space (Pauli error rate) & Dual space (Pauli fidelity) \\\hline
\Block{3-1}{Noisy Clifford\\gates\\(\cref{fig:CliffordIllu})} & \makecell{Pauli channels\\(with Clifford gates)} & \makecell{\(\displaystyle\sum_{P\in\mathcal{P}^m}p^P\mathcal{P}\mathcal{G}\)\\(\cref{eq:PCp})} & \makecell{\(\displaystyle\frac{1}{2^m}\sum_{Q\in\mathcal{P}^m}\lambda^Q|Q\rangle\!\rangle\langle\!\langle\mathcal{G}^\dagger(Q)|\)\\(\cref{eq:PCd})} \\\cline{2-4}
~ & learnability & \makecell{same as the Pauli fidelities\\(see \cref{sec:invar})} & \makecell{cycle space are learnable,\\cut space are not\\(Theorem 2 in \cite{chen2023learnability})} \\\cline{2-4}
~ & conversion formula & \makecell{\(\displaystyle p^P=\frac{1}{4^m}\sum_{Q\in\mathcal{P}^m}(-1)^{\langle P,Q\rangle}\lambda^Q\)\\(\cref{eq:PCd2p})} & \makecell{\(\displaystyle \lambda^Q=\sum_{P\in\mathcal{P}^m}(-1)^{\langle P,Q\rangle}p^P\)\\(\cref{eq:PCp2d})} \\\hline
\Block{3-1}{Noisy MCMs\\(\cref{fig:MCMIllu})} & \makecell{uniform stochastic\\instruments\\(with Clifford gates)} & \makecell{\(\displaystyle \sum_{\substack{a,b\in\mathbb{Z}_2^n\\P\in\mathcal{P}^m}}p_{a,b}^P(\mathcal{P}\otimes\lket{k+b}\lbra{k+a})\mathcal{G}\)\\(\cref{eq:USIp})} & \parbox{5cm}{\begin{align*}&\frac{1}{2^{2n+m}}\sum_{\substack{x,y\in\mathbb{Z}_2^n\\Q\in\mathcal{P}^m}}(-1)^{k\cdot(x+y)}\\&\lambda_{x,y}^Q|Q\otimes Z^y\rangle\!\rangle\langle\!\langle\mathcal{G}^\dagger(Q\otimes Z^x)|\end{align*}\\(\cref{eq:USId})} \\\cline{2-4}
~ & learnability & \makecell{different from the Pauli fidelities,\\partially characterized\\(see \cref{sec:pLearn})} & \makecell{cycle space are learnable,\\cut space are not\\(see \cref{sec:dLearn})} \\\cline{2-4}
~ & conversion formula & \parbox{5cm}{\begin{align*}p_{a,b}^P=&\frac{1}{4^{n+m}}\sum_{\substack{x,y\in\mathbb{Z}_2^n\\Q\in\mathcal{P}^m}}\\&(-1)^{a\cdot x+b\cdot y+\langle P,Q\rangle}\lambda_{x,y}^Q\end{align*}\\(\cref{eq:USId2p})} & \parbox{5cm}{\begin{align*}\lambda_{x,y}^Q=&\sum_{\substack{a,b\in\mathbb{Z}_2^n\\P\in\mathcal{P}^m}}\\&(-1)^{a\cdot x+b\cdot y+\langle P,Q\rangle}p_{a,b}^P\end{align*}\\(\cref{eq:USIp2d})} \\\hline
\Block{3-1}{Noisy measure-\\and-prepare\\(\cref{fig:MaPIllu})} & \makecell{measure and\\prepare instruments\\(with Clifford gates)} & \parbox{5cm}{\begin{align*} \sum_{\substack{a,b\in\mathbb{Z}_2^n\\P\in\mathcal{P}^m\\P_1\cdot P_2=P}}&q_a^{P_1}r_b^{P_2}\\&(\mathcal{P}\otimes\lket{k+b}\lbra{k+a})\mathcal{G}\end{align*}\\(\cref{eq:USIRp})} & \parbox{5cm}{\begin{align*}&\frac{1}{2^{2n+m}}\sum_{\substack{x,y\in\mathbb{Z}_2^n\\Q\in\mathcal{P}^m}}(-1)^{k\cdot(x+y)}\\&\zeta_x^Q\xi_y^Q|Q\otimes Z^y\rangle\!\rangle\langle\!\langle\mathcal{G}^\dagger(Q\otimes Z^x)|\end{align*}\\(\cref{eq:USIRd})} \\\cline{2-4}
~ & learnability & \makecell{same as uniform\\stochastic instruments,\\but overly constrained\\(see \cref{sec:application})} & \makecell{same as uniform\\stochastic instruments,\\but overly constrained\\(see \cref{sec:application})} \\\cline{2-4}
~ & conversion formula & \parbox{5cm}{\begin{align*}q_a^P=&\frac{1}{2^{n+2m}}\sum_{\substack{x\in\mathbb{Z}_2^n\\Q\in\mathcal{P}^m}}\\&(-1)^{a\cdot x+\langle P,Q\rangle}\zeta_x^Q\\r_b^P=&\frac{1}{2^{n+2m}}\sum_{\substack{x\in\mathbb{Z}_2^n\\Q\in\mathcal{P}^m}}\\&(-1)^{b\cdot x+\langle P,Q\rangle}\xi^Q\end{align*}\\(\cref{eq:USIRd2p})} & \parbox{5cm}{\begin{align*}\zeta_x^Q=&\sum_{\substack{a\in\mathbb{Z}_2^n\\P\in\mathcal{P}^m}}\\&(-1)^{a\cdot x+\langle P,Q\rangle}q_a^P\\\xi^Q_x=&\sum_{\substack{a\in\mathbb{Z}_2^n\\P\in\mathcal{P}^m}}\\&(-1)^{b\cdot x+\langle P,Q\rangle}r_b^P\end{align*}\\(\cref{eq:USIRp2d})} \\\hline
\end{NiceTabular}
}
\caption{Summary comparing physical space and dual space, integrating our findings with previous works.
Pointers are provided for relevant definitions and details.
The concept of cycle space and cut space will be discussed in~\cref{sec:PTG}.}
\label{tb:comp}
\end{table}

\subsection{Measure-and-prepare instruments}\label{sec:MPformal}
It turns out that the additional structure for measure-and-prepare instruments over uniform stochastic instruments is that the Pauli error rates $p$ factorizes according to $p_{a,b}^P=\sum_{P_1\cdot P_2=P}q_a^{P_1}r_b^{P_2}$, where $q$ and $r$ are two probability distributions and the Pauli fidelities factorizes according to  $\lambda'^Q_{x,y}=\zeta_x^Q\xi_y^Q$.
We defer the proofs to~\cref{sec:refresh}.

A comparison between the physical and dual space is listed in \cref{tb:comp}, which summarizes our results.
\FloatBarrier

\section{Algorithms for noise characterization}\label{sec:characterization}

\begin{figure}[t]
\centering
\begin{minipage}[b]{0.35\textwidth}
\subfloat[CB with interleaved gates]{\centering
\begin{tabular}{l}
    \Qcircuit @C=0.6em @R=0.4em {
    & \ctrl{1} & \qw & \gate{\phantom{I}} & \ctrl{1} & \qw & \gate{\phantom{I}} & \ctrl{1} & \qw & \gate{\phantom{I}} & \ctrl{1} & \qw & \gate{\phantom{I}} & \meter\\
    & \targ & \qw & \gate{\phantom{I}} & \targ & \qw & \gate{\phantom{I}} & \targ & \qw & \gate{\phantom{I}} & \targ & \qw & \gate{\phantom{I}} & \meter\\}\\
    \\
    1.~Prepare Pauli eigenstate\\
    2.~Apply repeated Clifford gates\\
    3.~Estimate Pauli observable\\
    4.~Obtain sum of log Pauli fidelities\\
    5.~Perform Fourier transform,\\~~~~obtain Pauli error
\end{tabular}
}
\end{minipage}\quad\quad
\begin{minipage}[b]{0.6\textwidth}
\subfloat[Our protocol]{\centering
\begin{tabular}{l}
\Qcircuit @C=0.6em @R=0.4em {
    & \ctrl{1} & \qw & \gate{\phantom{I}} & \ctrl{1} & \qw & \gate{\phantom{I}} & \ctrl{1} & \qw & \gate{\phantom{I}} & \ctrl{1} & \qw & \gate{\phantom{I}} & \meter\\
    & \targ & \meter & \gate{\phantom{I}} & \targ & \meter & \gate{\phantom{I}} & \targ & \meter & \gate{\phantom{I}} & \targ & \meter & \gate{\phantom{I}} & \meter\\
    & & \cwx & & & \cwx & & & \cwx & & & \cwx & & \\
}\\
\\
1.~Prepare Pauli eigenstate\\
    2.~Apply repeated MCMs\\
    3.~Estimate Pauli observable\\
    {\color{myblue}\textit{4.~Perform Fourier transform on measurement outcome}}\\
    5.~Obtain sum of log Pauli fidelities\\
    6.~Perform Fourier transform, obtain Pauli error
\end{tabular}
}
\end{minipage}
\caption{A comparison between CB with interleaved gates (a) and our protocol for characterizing MCMs (b). Here, blank gates represent single qubit Clifford gates. The quantum circuit in (a) is used to characterize a noisy $\mathrm{CNOT}$ gate, while the circuit in (b) is used to characterize a noisy MCM (measuring the target qubit after applying a $\mathrm{CNOT}$). The key difference between the two algorithms is an additional Fourier transform on measurement outcome (highlighted in blue).}
\label{fig:circuitExample}
\end{figure}

\subsection{Generalizing CB-type protocols}\label{sec:CBrecap}
For intuition, we will start with the special case of $n=0$ in our CB-type protocol.
In this scenario, our protocol simplifies to the CB with interleaved gates protocol for benchmarking noisy Clifford gates.
Variants of this protocol have been discussed in previous works, such as Refs.~\cite{flammia2020efficient,flammia2022averaged,carignan2023error}, with the exact version employed in Ref.~\cite{chen2023learnability} and utilized in Sec. SV of Ref.~\cite{van2023probabilistic}.
However, given that this is a special case, we will not dig into a comprehensive treatment here.

In CB with interleaved gates, a noisy Clifford gate to be characterized is modeled as an ideal Clifford gate followed by a Pauli noise channel.
Naturally we would want to learn the Pauli error rates of the error.
However, concatenation of Pauli channels corresponds to convolution of Pauli error rates, making it complicated to work with.
Therefore, a Fourier transformation is employed and the theory evolves within the dual space.
In the dual space, combine \cref{eq:Pcdng} to a Clifford gate $\mathcal{G}$ we get the noisy version of it,
\begin{equation}\label{eq:PCd}
\widetilde{\mc G}=\Lambda\circ\mc G=\frac{1}{2^m}\sum_{Q\in\mathcal{P}^m}\lambda^Q|Q\rangle\!\rangle\langle\!\langle\mathcal{G}^\dagger(Q)|.
\end{equation}
That is, the $\mathcal{G}^\dagger(Q)$ component of the input density matrix is converted to the $Q$ component of the output density matrix and shrunk by a factor of $\lambda^Q$ at the same time.
CB with interleaved gates repeatedly applies such noisy gates for multiple times and measures the component in the final density matrix.
In addition, single-qubit Clifford gates are interleaved between the mid-circuit measurements to ensure the correct concatenation of the transitions.
Take the simplest case where there are no interleaving single qubit Clifford gates as an example.
Start with a noisy input state $\frac{I+\lambda_S^PP}{2^m}$ ($\lambda_S^P$ is an unknown state preparation noise parameter), repeatedly applying the noisy Clifford gate $l$ times gives
\begin{equation}
\widetilde{\mathcal{G}}^l\left(\frac{I+\lambda_S^PP}{2^m}\right)=\frac{I+\lambda_S^P\lambda^{\mathcal{G}(P)}\cdots\lambda^{\mathcal{G}^l(P)}\mathcal{G}^l(P)}{2^m}.
\end{equation}
If we then measure (with noise) the observable $\mathcal{G}^l(P)$, the expected value we get will be $\lambda_S^P\lambda^{\mathcal{G}(P)}\cdots\lambda^{\mathcal{G}^l(P)}\lambda_M^{\mathcal{G}^l(P)}$, where $\lambda_M^{\mathcal{G}^l(P)}$ is an unknown measurement noise.
Suppose $P$ is invariant under $\mc G$, the expectation value becomes $\lambda_S^P(\lambda^P)^l\lambda_M^P$. In this case, by conducting experiments with different $l$ and conduct an exponential fitting, one can estimate $\lambda^P$ independently of the SPAM noise. For generic $P$ that is not invariant under $\mc G$, since $\{\mathcal{G}(P),\mathcal{G}^2(P),\ldots\}$ is periodic, similar approaches can be used to estimate certain product of Pauli fidelities SPAM robustly.

Inspired by this, we interleave single qubit Clifford gates between noisy MCMs.
An illustrative example circuit is presented in \cref{fig:circuitExample}.
However, in our MCM scenario, Pauli operators cease to be eigenvectors (cf. \cref{lem:unif} and note the summation over $y$ therein).
In fact, prior to the characterization, the eigenvectors are unknown, so we cannot prepare and measure the eigenvectors as in the CB with interleaved gates case.
Hence we still choose to prepare and measure Pauli observables, but now simple interleaving is insufficient because multiple transitions occur simultaneously, superposing together and yielding complicated results.
In detail, if we perform many experiments, group the terminating measurement results according to the MCM results and average them, the corresponding conditional expectation values are compromised of many terms, each term is a product of Pauli fidelities.
This new phenomenon is illustrated in \cref{fig:virtSpaceComp}.
\input{bigTikz}

Details of the figure may be understood more easily by referencing \cref{sec:protocol}.
Each column of nodes corresponds to a vectorization of a Pauli operator.
An arrow from column $P$ to column $Q$ represents $|Q\rangle\!\rangle\langle\!\langle P|$ together with the associated coefficients which are Pauli fidelities up to signs.
In \cref{fig:PCcompose}, arrows start at $\mathcal{G}^\dagger(Q)$ and end at $Q$ for $Q\in\mathcal{P}^m$.
Take $Q=IZ$ for example, we have $\mathcal{G}^\dagger(Q)=\mathrm{CNOT}^\dagger\cdot IZ\cdot \mathrm{CNOT}=ZZ$ and $Q=IZ$, so we get arrows from $ZZ$ to $IZ$.
Similarly, in \cref{fig:USIcompose}, arrows start at $\mathcal{G}^\dagger(Q\otimes Z^x)$ and end at $Q\otimes Z^y$ for $x,y\in\mathbb{Z}_2^n$ and $Q\in\mathcal{P}^m$.
Take $Q=I$, $x=1$, $y\in\{0,1\}$ for example, we have $\mathcal{G}^\dagger(Q\otimes Z^x)=\mathrm{CNOT}^\dagger\cdot IZ\cdot \mathrm{CNOT}=ZZ$, $Q\otimes Z^y\in\{II,IZ\}$, so we get arrows from $ZZ$ to $II$ and $IZ$.
Arrows of a specific row represent all elements of $\widetilde{\mathcal{G}}$ or $\mathcal{T}$ of a specific time.
By distributive law, the composition (the product) contains all possible combinations of edges from the channels.
However, since the vectorized Pauli operators are pairwise orthogonal, only consecutive edges (paths) may have non-zero contribution to the product since they correspond to the product of the form $|R\rangle\!\rangle\langle\!\langle Q|Q\rangle\!\rangle\langle\!\langle P|\cdots$, with coefficients omitted.
Some of the edges cannot form consecutive paths, indicating that their contribution is completely eliminated by subsequent randomized compiling, so if we want to learn them we should change the interleaving single qubit Clifford gates to make them consecutive.
If we specify a Pauli component of the input state and a terminating measurement, only paths of the specific starting point and end point will have contribution to the probabilities.

To disentangle this superposition of different transition paths so as to extract a single term, we need to apply another Fourier transformation on the measurement results in our data processing procedure.
It involves aggregating the averages while incorporating the Fourier coefficients.
This is one of our new ideas, and further details will be provided in \cref{sec:protocol}.

\subsection{Pattern transfer graph}\label{sec:PTG}
Before formally introducing our protocol, we first need to define the \emph{pattern transfer graph}, which is a directed graph $(V,E)$ where vertices $V=\mathbb{Z}_2^{n+m}$ corresponds to Pauli weight patterns and edges $E=\{e_{x,y}^Q:=(pt(\mathcal{G}^\dagger(Q\otimes Z^x)),pt(Q\otimes Z^y))|x,y\in\mathbb{Z}_2^n,Q\in\mathcal{P}^m\}$ corresponds to log Pauli fidelities.
An example of the Pattern transfer for $G=\mathrm{CNOT}$ ($m=n=1$) is shown in \cref{fig:pt}.
Our graph can be viewed as a generalization of the pattern transfer graph defined in~\cite{chen2023learnability}. For comparison, the pattern transfer graph for $G=\mathrm{CNOT}$ in the CB with interleaved gates case ($m=2$, $n=0$) is also included\footnote{
Note that, our labeling convention is slightly different from Ref.~\cite{chen2023learnability} because in our model, the noise happens after the Clifford gate, while in Ref.~\cite{chen2023learnability} the noise happens before the Clifford gate.
As a result, the pattern transfer graph defined in Ref.~\cite{chen2023learnability} is $(V,E)$ with $V=\mathbb{Z}_2^m$ and $E=\{e^Q:=(pt(Q),pt(\mathcal{G}(Q))|Q\in\mathcal{P}^m\}$, while our pattern transfer graph is defined as $(V,E)$ with $V=\mathbb{Z}_2^{n+m}$ and $E=\{e_{x,y}^Q:=(pt(\mathcal{G}^\dagger(Q\otimes Z^x)),pt(Q\otimes Z^y))|x,y\in\mathbb{Z}_2^n,Q\in\mathcal{P}^m\}$.
}.

\begin{figure}[ht]
\centering
\begin{tikzpicture}
[on grid,
pt/.style={circle, draw, thick, inner sep=0pt,minimum size=8mm, scale=1},
el/.style = {inner sep=4pt, align=left, scale=1}]
\node[pt] (00)               {$00$};
\node[pt] (01) [right=2.5cm of 00] {$01$};
\node[pt] (10) [below=2.5cm of 00] {$10$};
\node[pt] (11) [below=2.5cm of 01] {$11$};
\path[->] (00) edge node[el, auto] {$e^I_{0,1}$} (01);
\path[->] (01) edge node[el, above, pos=0.8] {$e^Z_{1,0}$} (10);
\path[->] (10) edge[bend left] node[el, below] {$e^Z_{0,1}$} (11);
\path[->] (11) edge node[el, below, pos=0.8] {$e^I_{1,0}$} (00);
\path[->] (11) edge[bend left] node[el, auto] {$e^X_{0,0},e^X_{1,0},e^Y_{0,0},e^Y_{1,0}$} (10);
\path[->] (11) edge[bend left] node[el, right] {$e^I_{1,1}$} (01);
\path[->] (01) edge[bend left] node[el, auto] {$e^Z_{1,1}$} (11);
\path[->] (00) edge[loop left] node[el, auto] {$e^I_{0,0}$} (00);
\path[->] (10) edge[loop left] node[el, auto] {$e^Z_{0,0}$} (10);
\path[->] (11) edge[loop right] node[el, auto, align=center] {$e^X_{0,1},e^X_{1,1}$\\$e^Y_{0,1},e^Y_{1,1}$} (11);
\draw [dashed] (5cm,-4cm) -- (5cm,1cm);
\node[pt] (c00) [right=5cm of 01]  {$00$};
\node[pt] (c01) [right=2.5cm of c00] {$01$};
\node[pt] (c10) [below=2.5cm of c00] {$10$};
\node[pt] (c11) [below=2.5cm of c01] {$11$};
\path[->] (c00) edge[loop left] node[el, auto] {$e^{II}$} (c00);
\path[->] (c01) edge[loop right] node[el, auto] {$e^{IX}$} (c01);
\path[->] (c10) edge[loop left] node[el, auto] {$e^{ZI}$} (c10);
\path[->] (c11) edge[loop right] node[el, auto] {$e^{XZ},e^{ZX},$\\$e^{YY},e^{XY},e^{YZ}$} (c11);
\path[->] (c11) edge[bend left] node[el, right] {$e^{IZ},$\\$e^{IY}$} (c01);
\path[->] (c01) edge[bend left] node[el, right] {$e^{ZY},$\\$e^{ZZ}$} (c11);
\path[->] (c11) edge[bend left] node[el, below] {$e^{XI},e^{YI}$} (c10);
\path[->] (c10) edge[bend left] node[el, below] {$e^{YX},e^{XX}$} (c11);
\end{tikzpicture}
\caption{left: Pattern transfer graph for $G=\mathrm{CNOT}$. Here $m=n=1$. The ancilla is the target qubit. This is the case for $X$ measurements using ancilla. As an example, the start point for $e_{0,1}^I$ is $pt(\mathcal{G}^\dagger(II))=00$, and the end point for $e_{0,1}^I$ is $pt(IZ)=01$.
right: Pattern transfer graph for $G=\mathrm{CNOT}$. Here $m=2$, $n=0$. No qubits are measured, so the situation degenerates into the CB with interleaved gates case. This graph can be found in~\cite{chen2023learnability}. Note that the rules are slightly different so the labels are different. See the footnote for explanation.}
\label{fig:pt}
\end{figure}

With pattern transfer graph in hand, we need some tools from graph theory~\cite{harary2018graph,bollobas1998modern} to exploit its power.
A (directed) path is an alternating sequence of vertices and edges, $v_0,e_1,v_1,\ldots,e_n,v_n$ such that $e_i=(v_{i-1},v_i)$.

Let $\epsilon_i\in\mathbb{R}$.
A \emph{$0$-chain} is a formal linear combination of vertices $\sum\epsilon_iv_i$, while a \emph{$1$-chain} is a formal linear combination of edges $\sum\epsilon_ie_i$.
For example, $-\,\Circled{01}+2\,\Circled{10}-\,\Circled{11}$ is a $0$-chain and $e^Z_{1,1}+2e^X_{0,0}$ is a $1$-chain (in this section, all graph theory examples are considered for \cref{fig:pt} left).
We emphasize that for pattern transfer graph, vertices are elements in $\mathbb{Z}_2^{n+m}$ (Pauli weight patterns), but the readers shall not confuse $0$-chain with addition in $\mathbb{Z}_2^{n+m}$ as we never perform addition on Pauli weight patterns in $\mathbb{Z}_2^{n+m}$ sense.
In the previous example of $0$-chain, we have circled the vertices for the distinction.
The \emph{edge space} $C$ is the vector space over $\mathbb{R}$ formed by $1$-chains, together with inner product $\langle\sum\epsilon_ie_i,\sum\epsilon_i'e_i\rangle=\sum\epsilon_i\epsilon_i'$.
For example, $e^Z_{1,1}+2e^X_{0,0}\in C$ and $\langle e^Z_{1,1}+2e^X_{0,0},e^I_{1,1}+e^Z_{1,1}\rangle=1$.

The \emph{boundary operator} $\partial$ is a linear operator that sends $1$-chains to $0$-chains such that if $e=(u,v)$, then $\partial e=v-u$.
For example, $\partial(e^Z_{1,1}+2e^X_{0,0})=(\Circled{11}-\,\Circled{01})+2(\Circled{10}-\,\Circled{11})=-\,\Circled{01}+2\,\Circled{10}-\,\Circled{11}$.
The \emph{coboundary operator} $\delta$ is a linear operator that sends $0$-chains to $1$-chains such that $\delta(v)=\sum\epsilon_ie_i$ where $\epsilon_i=1$ if $e_i=(u,v)$ for some $u\neq v$, $\epsilon_i=-1$ if $e_i=(v,u)$ for some $u\neq v$, and $\epsilon_i=0$ otherwise.
For example, $\delta(\Circled{00})=e^I_{1,0}-e^I_{0,1}$.
A \emph{cycle vector} is a $1$-chain with boundary $0$.
The \emph{cycle space} $Z$ is the subspace of $C$ formed by all cycle vectors.
For example, $\partial(e^I_{1,1}+e^Z_{1,1})=(\Circled{01}-\,\Circled{11})+(\Circled{11}-\,\Circled{01})=0$, so $e^I_{1,1}+e^Z_{1,1}\in Z$.
A \emph{cut vector} is a coboundary of some $0$-chain.
The \emph{cut space} $U$ is the subspace of $C$ formed by all cut vectors.
For example, $\delta(\Circled{00})=e^I_{1,0}-e^I_{0,1}\in U$.
The following lemma describes the relation between $C$, $Z$, and $U$.
\begin{lemma}[\cite{bollobas1998modern}, Sec. II.3, Theorem 9]\label{lem:orth}
The edge space $C$ is the orthogonal direct sum of cycle space $Z$ and cut space $U$.
\end{lemma}

\subsection{Protocol Details}\label{sec:protocol}
Now we present our protocol in details.
For this we make use of the pattern transfer graph.
For any given path $v_0,e_{x_1,y_1}^{Q_1},v_1,\ldots,e_{x_l,y_l}^{Q_l},v_l$, since $pt(Q_i\otimes Z^{y_i})=v_i=pt(\mathcal{G}^\dagger(Q_{i+1}\otimes Z^{x_{i+1}}))$, there exists $ H_i\in\mathcal{C}^{\otimes n+m}$ (recall that $\mathcal{C}$ is the group of single qubit Clifford gates, so $H_i$s are tensor products of single qubit Clifford gates) such that
\begin{equation}\label{eq:determineH}
\mathcal{H}_i(Q_i\otimes Z^{y_i})=\mathcal{G}^\dagger(Q_{i+1}\otimes Z^{x_{i+1}}).
\end{equation}
We claim that the following protocol is able to estimate
\begin{equation}
\log\left(\lambda_M^{v_l}\langle\!\langle\mathcal{G}^\dagger(Q_1\otimes Z^{x_1})|\rho\rangle\!\rangle\prod_{i=1}^l\lambda_{x_i,y_j}^{Q_i}\right) \equiv \log\langle\!\langle\mathcal{G}^\dagger(Q_1\otimes Z^{x_1})|\rho\rangle\!\rangle+\sum_{i=1}^l\log\lambda_{x_i,y_i}^{Q_i}+\log\lambda_M^{v_l}.
\end{equation}
\begin{enumerate}
\item Prepare Pauli eigenstate: Prepare an arbitrary but fixed state $\rho$. Preferably, it has a large overlap with the Pauli eigenstate $\frac{I+Q_1\otimes Z^{x_1}}{2^{n+m}}$.
\item Apply repeated MCMs: For $i=1\cdots l$, perform a compiled MCM (we need to apply randomized compiling on the MCM to ensure that it follows our noise model, see \cref{app:MCMrc} for details), record the result $m_i$, and then apply gate $H_i$ (No need to apply gate $H_i$ for $i=l$).
\item Estimate Pauli observable: Perform a compiled terminating measurement (again we need to apply randomized compiling on the terminating measurement to ensure that it follows our noise model, see \cref{app:TMrc} for details) for observable $Q_l\otimes Z^{y_l}$ and record the result $r$.
\item Perform Fourier transform on measurement outcome: Repeat the above procedure multiple times and estimate the expectation value $s= \mathbb E\left[(-1)^{\sum_{i=1}^lm_i\cdot(x_i+y_i)}r\right]$.
\item Obtain sum of log Pauli fidelities: output $\log s$.
\end{enumerate}

\cref{fig:prot} illustrates the circuit used in the protocol.
It turns out to be mathematically convenient if we modify the protocol a little bit by adding another experiment to cancel out the state preparation noise.
The modified protocol can estimate
\begin{equation}
\log\left(\frac{\lambda_M^{v_l}}{\lambda_M^{v_0}}\prod_{i=1}^l\lambda_{x_i,y_j}^{Q_i}\right) \equiv -\log\lambda_M^{v_0}+\sum_{i=1}^l\log\lambda_{x_i,y_i}^{Q_i}+\log\lambda_M^{v_l}.
\end{equation}
For the modification, simply replace step $5$ by
\begin{enumerate}
\setcounter{enumi}{4}
\item$\!\!\!\!\!^*$ Prepare $\rho$ and perform compiled terminating measurement for observable $\mathcal{G}^\dagger(Q_1\otimes Z^{x_1})$. Record the result as $r'$. Repeat multiple times and estimate the expectation value $t=\mathbb E[r']$.
\setcounter{enumi}{5}
\item$\!\!\!\!\!^*$ Output $\log\frac{s}{t}$.
\end{enumerate}

We note that the $x_i$, $y_i$ and $Q_i$s are the input to the algorithm.
$H_i$s can then be calculated via \cref{eq:determineH}.
The circuits are then determined and executed non-adaptively, independent of the measurement outcomes observed.
The final output is calculated based on the measurement outcomes and the inputs.

\begin{figure}[ht]
\centering
\begin{tabular}{c}
\Qcircuit @C=1em @R=.7em {
    & & \mathcal{G}^\dagger(Q_1\!\otimes\!Z^{x_1}\!) & & Q_1\!\otimes\!Z^{y_1} & & \mathcal{G}^\dagger(Q_2\!\otimes\!Z^{x_2}\!) & & Q_2\!\otimes\!Z^{y_2} & & & \cdots & & & Q_l\!\otimes\!Z^{y_l} \\
    \push{\rule{0em}{0.2em}} \\
    & & \Big\downarrow & \lambda_{x_1,y_1}^{Q_1} & \Big\downarrow & & \Big\downarrow & \lambda_{x_2,y_2}^{Q_2} & \Big\downarrow & & & \cdots & & \lambda_{x_l,y_l}^{Q_l} & \Big\downarrow & \lambda_M^{v_l}\\
    & /^{m} \qw & \qw & \multimeasure{1}{\mathcal{T}} & \qw & \gate{\phantom{H_1}} & \qw & \multimeasure{1}{\mathcal{T}} & \qw & \gate{\phantom{H_2}} &\qw & \dots & & \multimeasure{1}{\mathcal{T}} & \qw & \multimeasure{1}{\metersymb}\\
    & /^{n} \qw & \qw & \ghost{\mathcal{T}} & \qw & \gate{\phantom{H_1}} & \qw & \ghost{\mathcal{T}} & \qw & \gate{\phantom{H_2}} & \qw & \dots & & \ghost{\mathcal{T}} & \qw & \ghost{\metersymb}\\
    & & \push{\rule{2.5em}{0em}} & \dstick{m_1} \cwx & & \dstick{H_1} & \push{\rule{2.5em}{0em}} & \dstick{m_2} \cwx & & \dstick{H_2} & & \dstick{\cdots} & & \dstick{m_l} \cwx & & \dstick{r} \cwx \gategroup{4}{6}{5}{6}{.7em}{--} \gategroup{4}{10}{5}{10}{.7em}{--}\\\\\\\\
}
\end{tabular}

\begin{tikzpicture}
[node distance=20em,
pt/.style={circle,draw,thick,inner sep=0pt,minimum size=2em, scale=1},
el/.style = {inner sep=5pt, align=left, scale=1}]
\node[pt] (0) at (-2,0) {$v_0$};
\node[pt] (1) at (1,0) {$v_1$};
\node[pt] (2) at (4.5,0) {$v_2$};
\node[pt,draw=none] (dots) at (5.5,0) {$\cdots$};
\node[pt] (l) at (7.5,0) {$v_l$};
\node[pt,draw=none] (ghost) at (9,0) {};
\path[->] (0) edge node[el, above, pos=0.6] {$e_{x_1,y_1}^{Q_1}$} (1);
\path[->] (1) edge node[el, above, pos=0.7] {$e_{x_2,y_2}^{Q_2}$} (2);
\path[->] (2) edge (dots);
\path[->] (dots) edge node[el, above, pos=0.6] {$e_{x_l,y_l}^{Q_l}$} (l);
\end{tikzpicture}
\caption{Main circuit used in the protocol and the corresponding `walk' on the pattern transfer graph.
We interleave MCMs with the precomputed $H_i$s and perform a terminating measurement in the end.
The $\mathcal{G}^\dagger(Q_1\otimes Z^{x_1})$ component of the input state is multiplied by $\lambda^{Q_1}_{x_1,y_1}$ and transformed into the $Q_1\otimes Z^{y_1}$ component after applying the first MCM.
This component is then twisted by $H_1$ into $\mathcal{G}^\dagger(Q_2\otimes Z^{x_2})$ and so on\ldots
In the end the component becomes $Q_l\otimes Z^{y_l}$ which is then estimated by the terminating measurement.
In terms of Pauli weights, this sequence corresponds exactly to the transversal of the nodes $v_0,\ldots,v_l$ along the specified path, with the Pauli fidelities (edges) acquired along the way.
Those unshown components of the evolving underlying density matrix do not affect the result, as they are either undetected, eliminated by randomized compiling, or averaged out during data processing.}
\label{fig:prot}
\end{figure}

In the following, we consider the modified version of the protocol.
First we prove that the output is indeed the desired quantity.
\begin{theorem}[Main result]\label{thm:prot}
For any path $v_0,e_{x_1,y_1}^{Q_1},v_1,\ldots,e_{x_l,y_l}^{Q_l},v_l$, the function $-v_0+\sum_{i=1}^le_{x_i,y_i}^{Q_i}+v_l$ is learnable.
\end{theorem}

\noindent We note that our protocol outputs the path with its two end points (SPAM errors).
When the provided path forms a directed cycle (i.e., $v_0=v_l$), one can immediately see that the output $\sum_{i=1}^l\log\lambda_{x_i,y_i}^{Q_i}$ does not contain SPAM error parameters since they have been canceled.
In such cases, one may also concatenate the directed cycle with itself $L$ times and run the protocol on the extended path $v_0,e_{x_1,y_1}^{Q_1},v_1,\ldots,e_{x_l,y_l}^{Q_l},v_0,e_{x_1,y_1}^{Q_1},v_1,\ldots,e_{x_l,y_l}^{Q_l},\ldots,v_0$.
(In fact, as mentioned in \cref{sec:CBrecap}, this is the usual method used in CB-type algorithms.)
Then the protocol will output $L\sum_{i=1}^l\log\lambda_{x_i,y_i}^{Q_i}$, from which one can retrieve the product of fidelities.
One can further perform the experiment for varying $L$ and run a regression to get the original product of fidelities.

\begin{proof}[Proof of Theorem~\ref{thm:prot}]
It suffices to prove the correctness of our protocol.
The probability of observing MCM results $m_1,\ldots,m_l$ is the trace of resulting unnormalized density matrix
\begin{equation}
p_{m_1,\ldots,m_l}=\langle\!\langle I|\mathcal{T}_{m_l}H_{l-1}\cdots H_1\mathcal{T}_{m_1}|\rho\rangle\!\rangle.
\end{equation}
Here $I$ means the identity matrix.
Conditioned on observing this result, the quantum state is
\begin{equation}
|\sigma_{m_1,\ldots,m_l}\rangle\!\rangle=\frac{1}{p_{m_1,\ldots,m_l}}\mathcal{T}_{m_l}H_{l-1}\cdots H_1\mathcal{T}_{m_1}|\rho\rangle\!\rangle.
\end{equation}
Hence by \cref{lem:term}, the expectation of $r$ conditioned on observing the measurement results is
\begin{equation}\label{eq:naiveConcat}
\operatornamewithlimits{\mathbb{E}}\left[r\middle|m_1,\ldots,m_l\right]=\lambda_M^{v_l}\langle\!\langle Q_l\otimes Z^{y_l}|\sigma_{m_1,\ldots,m_l}\rangle\!\rangle.
\end{equation}
Thus
\begin{align}
&\operatornamewithlimits{\mathbb{E}}\left[(-1)^{\sum_{i=1}^lm_i\cdot(x_i+y_i)}r\right]\\
=&\sum_{m_1,\ldots,m_l\in\mathbb{Z}_2^n}p_{m_1,\ldots,m_l}(-1)^{\sum_{i=1}^lm_i\cdot(x_i+y_i)}\operatornamewithlimits{\mathbb{E}}\left[r\middle|m_1,\ldots,m_l\right]\\
=&\sum_{m_1,\ldots,m_l\in\mathbb{Z}_2^n}\lambda_M^{v_l}(-1)^{\sum_{i=1}^lm_i\cdot(x_i+y_i)}\langle\!\langle Q_l\otimes Z^{y_l}|\mathcal{T}_{m_l}H_{l-1}\cdots H_1\mathcal{T}_{m_1}|\rho\rangle\!\rangle\\
=&\frac{\lambda_M^{v_l}}{2^{(2n+m)l}}\sum_{\substack{m_1,\ldots,m_l\in\mathbb{Z}_2^n,a_1,\ldots,a_l\in\mathbb{Z}_2^n\\b_1,\ldots,b_l\in\mathbb{Z}_2^n,P_1,\ldots,P_l\in\mathcal{P}^m}}(-1)^{\sum_{i=1}^lm_i\cdot(x_i+y_i+a_i+b_i)}\prod_{i=1}^l\lambda_{a_i,b_i}^{P_i}\nonumber\\
&\langle\!\langle Q_l\otimes Z^{y_l}|P_l\otimes Z^{b_l}\rangle\!\rangle\langle\!\langle \mathcal{G}^\dagger(P_l\otimes Z^{a_l})|\mathcal{H}_{l-1}(P_{l-1}\otimes Z^{b_{l-1}})\rangle\!\rangle\cdots\langle\!\langle \mathcal{G}^\dagger(P_1\otimes Z^{a_1})|\rho\rangle\!\rangle\label{eq:channelsubstitute}\\
=&\lambda_M^{v_l}\sum_{\substack{a_1,\ldots,a_l\in\mathbb{Z}_2^n\\P_1,\ldots,P_l\in\mathcal{P}^m}}\prod_{i=1}^l\lambda_{a_i,x_i+y_i+a_i}^{P_i}f(a_1,\ldots,a_l,P_1,\ldots,P_l)\langle\!\langle\mathcal{G}^\dagger(P_1\otimes Z^{a_1})|\rho\rangle\!\rangle,
\end{align}
where function $f$ takes value $0,\pm1$ depending on the inner products
\begin{align}
&2^{(n+m)l}f(a_1,\ldots,a_l,P_1,\ldots,P_l)\nonumber\\
=&\langle\!\langle Q_l\otimes Z^{y_l}|P_l\otimes Z^{x_l+y_l+a_l}\rangle\!\rangle\langle\!\langle\mathcal{G}^\dagger(P_l\otimes Z^{a_l})|\mathcal{H}_{l-1}(P_{l-1}\otimes Z^{x_{l-1}+y_{l-1}+a_{l-1}})\rangle\!\rangle\\
&\cdots\langle\!\langle\mathcal{G}^\dagger(P_2\otimes Z^{a_2})|\mathcal{H}_1(P_1\otimes Z^{x_1+y_1+a_1})\rangle\!\rangle.\nonumber
\end{align}
\cref{eq:channelsubstitute} uses \cref{lem:unif}.


Next, we prove that only one term is left in the summation, that is, $f(a_1,\ldots,a_l,P_1,\ldots,P_l)=\prod_{i=1}^l\delta_{a_i,x_i}\delta_{P_i,Q_i}$.
Here, $\delta$ is the Kronecker delta function that takes value $1$ only when its two subscripts are equal, and takes value $0$ otherwise.
We prove this inductively.
If $f$ is non-zero, then from the first inner product we have $P_l=Q_l$, $a_l=x_l$.
Since $\mathcal{H}_{l-1}(Q_{l-1}\otimes Z^{y_{l-1}})=\mathcal{G}^\dagger(Q_l\otimes Z^{x_l})$, from the second inner product we have $P_{l-1}=Q_{l-1}$, $a_{l-1}=x_{l-1}$\dots
In the end we have $P_1=Q_1$, $a_1=x_1$.
Hence $f$ is non-zero only when $P_i=Q_i$ and $a_i=x_i$, and it takes value $1$ in this case.
Putting this result back we get
\begin{equation}
\operatornamewithlimits{\mathbb{E}}\left[(-1)^{\sum_{i=1}^lm_i\cdot(x_i+y_i)}r\right]=\lambda_M^{v_l}\langle\!\langle\mathcal{G}^\dagger(Q_1\otimes Z^{x_1})|\rho\rangle\!\rangle\prod_{i=1}^l\lambda_{x_i,y_i}^{Q_i}.
\end{equation}
This is the expected value of $s$.
On the other hand, by \cref{lem:term} the expected value of $t$ is $\lambda_M^{v_0}\langle\!\langle\mathcal{G}^\dagger(Q_1\otimes Z^{x_1})|\rho\rangle\!\rangle$, hence $\log\frac{s}{t}$ is an estimate of the desired quantity.
\end{proof}

\section{Learnability of MCMs}\label{sec:learnability}
We have seen in Theorem~\ref{thm:prot} that certain combinations of Pauli fidelities can be learned using our generalized CB algorithm. Now we show that, under our noise model these turn out to be all information that can be SPAM-robustly learned about noisy MCMs \emph{via any algorithm}. For this purpose, we develop a theory on the learnability of MCMs, generalizing the framework from Ref.~\cite{chen2023learnability} about the learnability of noisy Clifford gates.

\medskip
To start with, let us formally summarize our noise model assumptions:
\begin{enumerate}
    \item All single-qubit gates can be noiselessly implemented.
    \item A set of multi-qubit Clifford gates $\{\mc G_i\}$ can be implemented followed by (unknown) gate-dependent Pauli noise channels $\{\Lambda_{\mc G_{i}}\}$.
    \item A noisy MCM $\{\mc T_k\}$ with (unknown) Pauli fidelities $\{\lambda_{x,y}^{Q}\}$ can be applied.
    \item An unknown but fixed initial state $\rho$ can be prepared.\footnote{We note that regarding the initial state, though it does not affect our learnability theory, for our protocol the overlap $\langle\!\langle\mathcal{G}^\dagger(Q_1\otimes Z^{x_1})|\rho\rangle\!\rangle$ do affect the sample complexity. It is desirable to choose and prepare an initial state with high overlap $\langle\!\langle\mathcal{G}^\dagger(Q_1\otimes Z^{x_1})|\rho\rangle\!\rangle$.}
    \item Any POVM can be measured following an unknown (symmetric) Pauli noise channel $\Lambda_M$. See Appendix~\ref{app:TMrc} for details.
    \item All Pauli fidelities of noise channels are strictly positive. All noise channels are not at the boundary of the completely-positive polytope.
\end{enumerate}
Assumption 1 is standard for randomized-compiling-based protocols~\cite{wallman2016noise,hashim2020randomized,beale2023randomized}, and can be relaxed such that all layers of single-qubit gates have \emph{gate-independent noise}. Assumptions 2-4 can be enforced using randomized compiling. 
Note that, if we exclude assumption 3, the noise model reduces to the standard Pauli noise model with Clifford gates, for which the learnability has been studied in Ref.~\cite{chen2023learnability}.
We also remark that, it is straightforward to generalize our learnability theory to allow multiple distinct noisy MCM gadgets. We omit that for conciseness.
Assumption 6 is mostly for mathematical convenience: the first part basically says that the noise is not overwhelmingly large; the second part ensures that all noise parameters can be perturbed without making the noise model nonphysical.

To learn the noise parameters in the noise model $\mc N = \{\rho, \{\Lambda_{\mc G_i}\}, \{\lambda_{x,y}^Q\},\Lambda_M\}$, the most general form of experiments one can perform is to prepare the initial state, apply a sequence of gates and MCM gadgets, and perform a terminating measurement.
Any experiment maps a (realization of the) noise model to a probability distribution over the measurement outcomes (from both the MCMs and the terminating measurements).
We say two noise models $\mathcal{N}_1,\mathcal{N}_2$ are \emph{indistinguishable} if for every possible experiment, they yield the same probability distribution. 
Otherwise they are \emph{distinguishable}.

We are interested in which parameters of the noise model are learnable from experiments. Formally, a function $f$ on noise models maps a noise model $\mc N$ to a real number, denoted as $f(\mc N)$.
For example, a $0$-chain can be viewed as a function of noise models that reflects terminating measurement noises $\left(\sum\epsilon_iv_i\right)(\mathcal{N})=\sum\epsilon_i\log\lambda_M^{v_i}$.
Similarly, a $1$-chain can be viewed as a function of noise models that reflects MCM noises $\left(\sum\epsilon_{x,y}^Qe_{x,y}^Q\right)(\mathcal{N})=\sum\epsilon_{x,y}^Q\log\lambda_{x,y}^Q$.

A function $f$ is called \emph{learnable}~\cite{chen2023learnability} if
\[\forall\mathcal{N}_1,\mathcal{N}_2:~
\mathcal{N}_1~\text{is indistinguishable from}~\mathcal{N}_2~\Rightarrow~f(\mathcal{N}_1)= f(\mathcal{N}_2).\]
Otherwise, $f$ is \emph{unlearnable}. This definition is as expected, because the ability to learn an unlearnable function would imply the ability to distinguish indistinguishable noise models, which leads to a contradiction.
From this definition alone, $f$ being learnable is only a necessary condition for the existence of an experiment to actually learn its value.
Here we will prove that our learning algorithm can indeed learn any of such learnable functions within arbitrary precision. We also remark that a function being unlearnable is a fundamental limitation for \emph{any learning protocols}, not just specific to CB-type protocols.

\subsection{Learnability of Pauli fidelities}\label{sec:dLearn}
The protocol presented in \cref{sec:protocol} is important in helping us understand the learnability of Pauli fidelities.
To make a complete characterization of the learnable information, first we need the following lemma.
\begin{lemma}[Lemma 1, Supplementary of Ref.~\cite{chen2023learnability}]
Denote the set of all learnable $1$-chains by $F_L$. Then $F_L$ forms a linear subspace of the edge space $C$.
\end{lemma}
The lemma states that learnability defined above is closed under linear operations.
For completeness we present a proof here.
\begin{proof}
Given any $\mu_1,\mu_2\in F_L$, $\forall\mathcal{N}_1,\mathcal{N}_2$, $\forall\alpha\neq0$,
\begin{equation}
\begin{aligned}
\mathcal{N}_1~\text{is indistinguishable from}~\mathcal{N}_2&\Rightarrow
\mu_1(\mathcal{N}_1)=\mu_1(\mathcal{N}_2)\text{ and }\mu_2(\mathcal{N}_1)=\mu_2(\mathcal{N}_2)\\
&\Rightarrow(\mu_1+\alpha\mu_2)(\mathcal{N}_1)=(\mu_1+\alpha\mu_2)(\mathcal{N}_2).
\end{aligned}
\end{equation}
Thus $\mu_1+\alpha\mu_2\in F_L$. Note that the last line uses the linearity of $1$-chains.
\end{proof}

\begin{corollary}\label{cor:chainl}
For any $1$-chain $\mu$, $\mu+\partial\mu$ is learnable.
\end{corollary}
\begin{proof}
As a special case of \cref{thm:prot}, for any edge $e$, $e+\partial e$ is learnable.
Hence by linearility, $\mu+\partial\mu$ is learnable.
\end{proof}

\begin{corollary}\label{cor:cycl}
$1$-chains in the cycle space $Z$ are learnable.
\end{corollary}
\begin{proof}
For $1$-chains $\mu\in Z$, by \cref{cor:chainl}, $\mu+\partial\mu$ is learnable.
By definition of cycles, $\partial\mu=0$, thus $\mu$ is learnable.
\end{proof}
We remark that our focus is on MCMs ($1$-chains), and we are not interested in terminating measurement errors ($0$-chains).
Thus the implication of \cref{cor:cycl} is that cycle space can be learned unaffected by state preparation and terminating measurement noises, that is, SPAM robustly.

Furthermore, we will show that the cycle space is all the information that can be learned SPAM robustly, and thus our protocol learns all the information that can be learned.
\begin{theorem}\label{thm:lear}
The protocol in \cref{thm:prot} is complete in the sense that the space of learnable information $F_L$ is equal to the cycle space $Z$.
\end{theorem}
\begin{proof}
$Z\subseteq F_L$ is already proved in \cref{cor:cycl}, so it remains to show that $F_L\subseteq Z$.
By \cref{lem:orth}, it suffices to show that $F_L$ is orthogonal to the cut space $U$.
We will show that every cut vector induces a gauge transformation that convert one noise model to another indistinguishable noise model. Since learnable functions should be invariant under such transformation by definition, we can conclude that they should be orthogonal to the cut space.

Recall that a generic experiment begins with an initial state $\rho$, followed by a sequence of MCMs interleaved by single-qubit or multi-qubit Clifford gates $\mathcal{C}_i$ (if we extend our noise model to include non-Clifford multi-qubit gates, they can appear in the circuit, too), and finally a terminating measurement. Assume that there are $d$ layers of MCMs with outcome denoted by $m_1,\cdots,m_d$, and that the outcome of the terminating measurement is denoted by $m_*$, the probability distribution of the outcomes is given by
\begin{equation}
    \Pr[m_1,\cdots,m_d,m_*] = \Tr(E_{m_*}\left(\mc C_d\circ \mc T_{m_d}^{} \circ\cdots\circ \mc C_1\circ\mc T_{m_1}^{}\circ\mc C_0(\rho)\right)).
\end{equation}
Here, each $\mc C_i$ can be a layer of noiseless single-qubit gates, or a noisy multi-qubit Clifford gates, or a concatenation of both. $E_{m^*}$ is the POVM element of the noisy terminating measurement.

Thus we can see that, for any invertible linear map $\mathcal{D}$, the following gauge transformation does not change the probability distribution.
\begin{equation}
\left\{
\begin{aligned}
&\rho\mapsto\mathcal{D}{(\rho)}\\
&\mathcal{C}_i\mapsto\mathcal{D}\mathcal{C}_i\mathcal{D}^{-1}\\
&\mathcal{T}_k\mapsto\mathcal{D}\mathcal{T}_k\mathcal{D}^{-1}\\
&E_j\mapsto(\mathcal{D}^{-1})^\dagger(E_j)
\end{aligned}
\right..
\end{equation}
For a noise model $\mc N_1$, we can use this to construct an indistinguishable model $\mathcal{N}_2$ from $\mathcal{N}_1$.

Given any cut vector $\delta(\nu)\in U$, since $\delta(\sum v)=0$, wlog assume $\nu=\sum\epsilon_{v_i}v_i$ has coefficient $0$ on vertex $\Circled{0_{m+n}}$.
Let $\mathcal{D}$ be the Pauli diagonal map defined by
\begin{equation}
\forall P\in\mathcal{P}^{n+m},~\mathcal{D}(P)=\eta^{\epsilon_{pt(P)}}P
\end{equation}
for a positive $\eta\neq1$.
We hence need to verify that $\mathcal{N}_2$ still satisfies all of our model assumptions:
\begin{enumerate}
    \item If $\mc C$ is a layer of single-qubit gate, it preserves the pattern of any input Pauli operator\footnote{
    More precisely, since $\mc C$ is allowed to be non-Clifford, it can map an input Pauli to multiple output Paulis, each of which has the same pattern as the input.
    }, thus $\mc C_i = \mc D\mc C_i\mc D^{-1}$, i.e., single-qubit gates remains noiseless.
    \item If $\mc C$ is a multi-qubit Clifford gate $\mc G_i$ followed by a Pauli noise channel $\Lambda_{\mc G_i}$, we have \begin{equation}
    \mc D\Lambda_{\mc G_i}\mc G_i\mc D^{-1} = \left(\mc D\Lambda_{\mc G_i}\mc G_i\mc D^{-1}\mc G_i^{\dagger}\right)\mc G_i = \Lambda'_{\mc G_i}\mc G_i,
    \end{equation}
    which is still the same Clifford gate followed by a Pauli channel.
    \item From \cref{lem:unif}, we can easily see that $\{\mc D\mc T_k\mc D^{-1}\}$ remains a compiled MCM (i.e., $\{\mc D\mc U_k\mc D^{-1}\}$ remains a uniform stochastic instrument), with Pauli fidelities changed according to
    \begin{equation}\label{eq:gauge_MCM}
    \begin{aligned}
    \log\lambda_{x,y}^Q& \mapsto \log\lambda_{x,y}^Q+(\epsilon_{pt(Q\otimes Z^y)}-\epsilon_{pt(\mc G^\dagger( Q\otimes Z^x))})\log\eta.\\
    &= \log\lambda_{x,y}^Q+\expval{\delta(\nu),e_{x,y}^Q}\log\eta
    \end{aligned}
    \end{equation}
    Note that since $e_{0,0}^I$ is always a self loop on vertex $0$, the trace preserving condition $\log\lambda_{0,0}^I=0$ is preserved.
    \item The initial state transforms from $\rho$ to $\mc D(\rho)$. For the terminating measurement, if the ideal POVM is $\{F_j\}$, its noisy implementation is $\{\Lambda_M^\dagger(F_j)\}$ where $\Lambda_M$ is a symmetric Pauli channel. The transformation results in $\Lambda_M\mapsto\Lambda_M\mc D^{-1}$ which is still a symmetric Pauli channel.
    \item Finally, it is not hard to see Assumption 5 still holds under the transformation, as long as $\eta$ is sufficiently close to $1$.
\end{enumerate}


Thus $\mathcal{N}_2$ satisfies all of our assumptions.
Moreover, for any $1$-chain $\mu$, thanks to Eq.~\eqref{eq:gauge_MCM} we have
\begin{equation}
\mu(\mathcal{N}_2)=\mu(\mathcal{N}_1)+\langle\delta(\nu),\mu\rangle\log\eta.    
\end{equation}
Hence if $\mu\in F_L$, by the definition of learnable functions, we must have $\langle\delta(\nu),\mu\rangle=0$.
This completes the proof.
\end{proof}

\subsection{Learnability of Pauli error rates}\label{sec:pLearn}
\cref{thm:lear} gives a complete classification for the learnable Pauli fidelities.
Since Pauli error rates have a clearer physical meaning, naturally we would ask for a classification for the learnable Pauli error rates.
In this subsection, we give a partial classification under low noise rate assumption.
Under low noise rate assumption, the Pauli fidelities are close to $1$, thus we can make approximation
\begin{align}
p_{a,b}^P&=\frac{1}{4^{n+m}}\sum_{x,y\in\mathbb{Z}_2^n,Q\in\mathcal{P}^m}(-1)^{a\cdot x+b\cdot y+\langle P,Q\rangle}\lambda_{x,y}^Q\\
&\approx\frac{1}{4^{n+m}}\sum_{x,y\in\mathbb{Z}_2^n,Q\in\mathcal{P}^m}(-1)^{a\cdot x+b\cdot y+\langle P,Q\rangle}(1+\log\lambda_{x,y}^Q)\label{eq:fidelityapprox}\\
&=\frac{1}{4^{n+m}}\sum_{x,y\in\mathbb{Z}_2^n,Q\in\mathcal{P}^m}(-1)^{a\cdot x+b\cdot y+\langle P,Q\rangle}\log\lambda_{x,y}^Q+\delta_{a,0}\delta_{b,0}\delta_{P,I}.\label{eq:errorrate}
\end{align}
\cref{eq:fidelityapprox} uses the fact that when $x\approx1$, $x\approx1+\log x$, and the $\delta$ in \cref{eq:errorrate} means the Kronecker delta function.

As an example, the learnable Pauli error rates for $n=m=1$ MCMs with $G=\mathrm{CNOT}$ are
\begin{gather*}
p_{1,1}^I,p_{1,1}^Z,p_{1,1}^X,p_{1,1}^Y\\
p_{0,0}^I+p_{0,0}^Z,p_{0,1}^I+p_{1,0}^Z+p_{0,1}^Z+p_{1,0}^I\\
p_{0,0}^X-p_{0,0}^Y,p_{1,0}^X-p_{1,0}^Y,p_{0,1}^X-p_{0,1}^Y,p_{0,1}^X-p_{0,1}^Z\\
p_{0,1}^I+p_{1,0}^I-p_{0,0}^I,p_{0,1}^X+p_{1,0}^X+p_{0,0}^X,p_{0,0}^I+p_{1,0}^I+p_{0,1}^X.
\end{gather*}
More generally, we have the following propositions.
Proofs are deferred to~\cref{sec:defer}.
\begin{proposition}\label{prop:1}
$\forall G,P$, $\forall a\neq0,b\neq0$, $p_{a,b}^P$ is learnable.
\end{proposition}

The proof is existential, but Corollary~\ref{cor:chainl} also suggests a concrete way to learn this quantity. Specifically, one can learn each $\log\lambda_{x,y}^Q$ as in Eq.~\eqref{eq:errorrate} together with its end point using our Protocol. Thanks to \cref{cor:chainl}, the end points will cancel as we sum up all the estimators. 
Alternatively, one can try to directly decompose Eq.~\eqref{eq:errorrate} into sum of cycles.
As a concrete example, we give the protocol for learning $p_{1,1}^I$ when $G=\mathrm{CNOT}$ (the setting for \cref{fig:pt}).
$p_{1,1}^I$ can be decomposed as
\begin{align}
&16p_{1,1}^I\approx\nonumber\\
&\log\lambda_{0,0}^I+(\log\lambda_{1,1}^X+\log\lambda_{1,1}^Y+\log\lambda_{0,0}^Y+\log\lambda_{0,1}^Z+\log\lambda_{0,0}^X+\log\lambda_{0,0}^Z+\log\lambda_{0,1}^Z+\log\lambda_{1,1}^I+\log\lambda_{1,1}^Z)\nonumber\\
&-(\log\lambda_{0,1}^Y+\log\lambda_{1,0}^X+\log\lambda_{0,1}^Z+\log\lambda_{0,1}^X+\log\lambda_{1,0}^Y+\log\lambda_{0,1}^Z+\log\lambda_{1,0}^I+\log\lambda_{0,1}^I+\log\lambda_{1,0}^Z+\log\lambda_{0,1}^Z).\label{eq:decompexp}
\end{align}
We note that the first term, $\log\lambda_{0,0}^I$ is fixed to be $0$, so we only need to apply our protocol $2$ times to learn the second and the last term.

An important application for MCMs is syndrome measurement.
In this case we measure the stabilizer $S$ by introducing an ancilla, perform a Hadamard gate on the ancilla, apply controlled-$S$ using the ancilla as control, and then apply another Hadamard transformation on the ancilla and measure the ancilla.
See \cref{fig:syndrome} for an illustration of the ideal Clifford gate $G$ in the case when we want to measure stabilizers $S_1,\ldots,S_n$ simultaneously.
Similar noise extraction under syndrome measurements settings have been studied in Refs.~\cite{wagner2022pauli,wagner2023learning}, though under different noise assumptions.

The following two propositions consider this case.
For $k\in\mathbb{Z}_2^n$, define $S^k=\prod S_i^{k_i}$.
Denote the set of Pauli operators that commute with all stabilizers by $C_{\mathcal{P}^m}(\mathcal{S})$.
\begin{figure}
\centering
\begin{tabular}{c}
\Qcircuit @C=1em @R=.7em {
    & /^m \qw & \qw & \gate{S_1} & \qw & \gate{S_2} & \qw & \qw & & & & /^m \qw & \multigate{2}{G} & \qw \\
    & \qw & \gate{H} & \ctrl{-1} & \gate{H} & \qw & \qw & \qw & & \!\!\!\!\!\!\Longleftrightarrow & & \qw & \ghost{G} & \qw \\
    & \qw & \qw & \qw & \gate{H} & \ctrl{-2} & \gate{H} & \qw & & & & \qw &\ghost{G} & \qw
}
\end{tabular}
\caption{Illustration of the ideal Clifford gate $G$ for syndrome measurements. Here we take $n=2$.}
\label{fig:syndrome}
\end{figure}

\begin{proposition}\label{prop:2}
$\forall P\in C_{\mathcal{P}^m}(\mathcal{S})$, $\sum_{k\in\mathbb{Z}_2^n}p_{0,0}^{S^kP}$ is learnable.
\end{proposition}

The summation over $k\in\mathbb{Z}_2^n$ can be intuitively understood as the following.
Since we are measuring stabilizers, it makes no effect if we randomly apply some stabilizers before the measurement.
Wlog suppose we indeed applied the random stabilizers.
Then errors that differ by stabilizers becomes equivalent, thus only their average error rates can be learnable.

We note that when $S_1=\cdots=S_n=I$, or equivalently, $G=I$ and we are implementing a subsystem measurement, this proposition shows the learnability of $p_{0,0}^I$, which is the process fidelity between the uniform stochastic instrument and the ideal subsystem measurement~\cite{mclaren2023stochastic}.

As another special case, when $n=0$, no stabilizers are measured.
By the requirements on $G$, $G$ must be identity.
Now $C_{\mathcal{P}^m}(\mathcal{S})=\mathcal{P}^m$, and this proposition degenerates into the fact that for an isolated Pauli channel, every parameter is learnable.

\begin{proposition}\label{prop:3}
$\forall P\in C_{\mathcal{P}^m}(\mathcal{S})$, $\forall a\neq0$, $\sum_{k\in\mathbb{Z}_2^n}p_{0,a}^{S^kP}+p_{a,0}^{S^kP}$ is learnable.
\end{proposition}

The summation over $k\in\mathbb{Z}_2^n$ originates from the random stabilizer intuition.
The intuitive meaning of summing $p_{0,a}^P$ and $p_{a,0}^P$ together is that we can not distinguish between wrong but consistent with post-measurement state result and a correct but inconsistent with post-measurement state result.
These two possibilities can only be learned as a unity.

\section{Application: testing the independence of measurement and state preparation}\label{sec:application}
In \cref{sec:MPformal} we described the measure and prepare instrument model.
The key assumption we made is that the measurement and the subsequent state preparation are independent and suffers from uncorrelated noises.
In experiment, we can test whether the assumption is true or not by checking the additional structures of the model.
To see why we can do this, we need more understanding on the structure.
\begin{lemma}\label{lem:indepequiv}
$\exists\zeta,\xi$ s.t. $\forall x,y\in\mathbb{Z}_2^n,Q\in\mathcal{P}^m$, $\lambda'^Q_{x,y}=\zeta_x^Q\xi_y^Q$ if and only if $\forall x_1,x_2,y_1,y_2\in\mathbb{Z}_2^n,Q\in\mathcal{P}^m$, the correlation $c^Q_{x_1,x_2,y_1,y_2}\coloneqq\log\lambda'^Q_{x_1,y_1}+\log\lambda'^Q_{x_2,y_2}-\log\lambda'^Q_{x_2,y_1}-\log\lambda'^Q_{x_1,y_2}=0$.
\end{lemma}
\begin{proof}
The only if part is easy.
Suppose $\exists\zeta,\xi$ s.t. $\forall x,y\in\mathbb{Z}_2^n,Q\in\mathcal{P}^m$, $\lambda'^Q_{x,y}=\zeta_x^Q\xi_y^Q$, then for $\forall x_1,x_2,y_1,y_2\in\mathbb{Z}_2^n,Q\in\mathcal{P}^m$,
\begin{align}
&\log\lambda'^Q_{x_1,y_1}+\log\lambda'^Q_{x_2,y_2}-\log\lambda'^Q_{x_2,y_1}-\log\lambda'^Q_{x_1,y_2}\\
=&\log\zeta_{x_1}^Q+\log\xi_{y_1}^Q+\log\zeta_{x_2}^Q+\log\xi_{y_2}^Q-\log\zeta_{x_2}^Q-\log\xi_{y_1}^Q-\log\zeta_{x_1}^Q-\log\xi_{y_2}^Q=0.
\end{align}
Conversely, for the if part, set $x_1=y_1=0$, $x_2=x$, $y_2=y$ we have
\begin{equation}
\log\lambda_{x,y}'^Q=\log\lambda_{x,0}'^Q+\log\lambda_{0,y}'^Q-\log\lambda_{0,0}'^Q.
\end{equation}
Thus we can set $\zeta_x^Q=\frac{\lambda_{x,0}'^Q}{\sqrt{\lambda_{0,0}'^Q}}$, $\xi_y^Q=\frac{\lambda_{0,y}'^Q}{\sqrt{\lambda_{0,0}'^Q}}$, proving the result.
\end{proof}
From \cref{lem:indepequiv} we see that to test the independence of measurement and state preparation, we can equivalently test whether the correlations are all zero.
Furthermore, these conditions are learnable information and hence are indeed verifiable.
\begin{lemma}\label{lem:indepLearnable}
$\forall x,y\in\mathbb{Z}_2^n,Q\in\mathcal{P}^m$, $\log\lambda'^Q_{x_1,y_1}+\log\lambda'^Q_{x_2,y_2}-\log\lambda'^Q_{x_2,y_1}-\log\lambda'^Q_{x_1,y_2}$ is learnable.
\end{lemma}
\begin{proof}
\begin{align}
\partial\left(e^Q_{x_1,y_1}+e^Q_{x_2,y_2}-e^Q_{x_2,y_1}-e^Q_{x_1,y_2}\right)=&-pt(\mathcal{G}^\dagger(Q\otimes Z^x_1))+pt(Q\otimes Z^y_1)-pt(\mathcal{G}^\dagger(Q\otimes Z^x_2))+pt(Q\otimes Z^y_2)\nonumber\\
&-pt(\mathcal{G}^\dagger(Q\otimes Z^x_2))+pt(Q\otimes Z^y_1)-pt(\mathcal{G}^\dagger(Q\otimes Z^x_1))+pt(Q\otimes Z^y_2)=0
\end{align}
\end{proof}
Since $c^Q_{x_1,x_2,y_1,y_2}=-c^Q_{x_2,x_1,y_1,y_2}=-c^Q_{x_1,x_2,y_2,y_1}$, in experiments we only need to estimate for $x_1<x_2$ and $y_1<y_2$ (comparison under alphabetical order).

\section{Numerical simulations}
\label{sec:numerical}
In this section, we perform numerical simulations for the simple example of $G=\mathrm{CNOT}$ and $m=n=1$, which is a basic gadget for syndrome measurement.

\paragraph{Learning Pauli fidelities.}\label{par:fid}
We randomly generate a noisy quantum instrument and simulate our (modified) protocol on it.
For each of the cycle basis, the way we learn it is to concatenate the cycle multiple times in the main circuit (Steps $1-4$) so that the cycle becomes of length $\ell=12$ (i.e. $12$ MCMs in each of the main circuit).
We assign values for the random variables in the randomized compilings $100$ times.
This results in $100$ random circuits, which we call compiled circuits, and for each compiled circuit we sample $100$ shots.
We take the same number of shots for the auxiliary circuit (Step $5$), so in total $N=20000$ shots are taken for each data point.
We then process the data to get the geometric mean of the Pauli fidelities in the cycle basis.
The standard deviations of the obtained results are calculated through bootstrapping.
Our results are shown in \cref{fig:learnedCycle}.

\begin{figure}[ht]
\centering
\includegraphics[width=0.99\textwidth]{"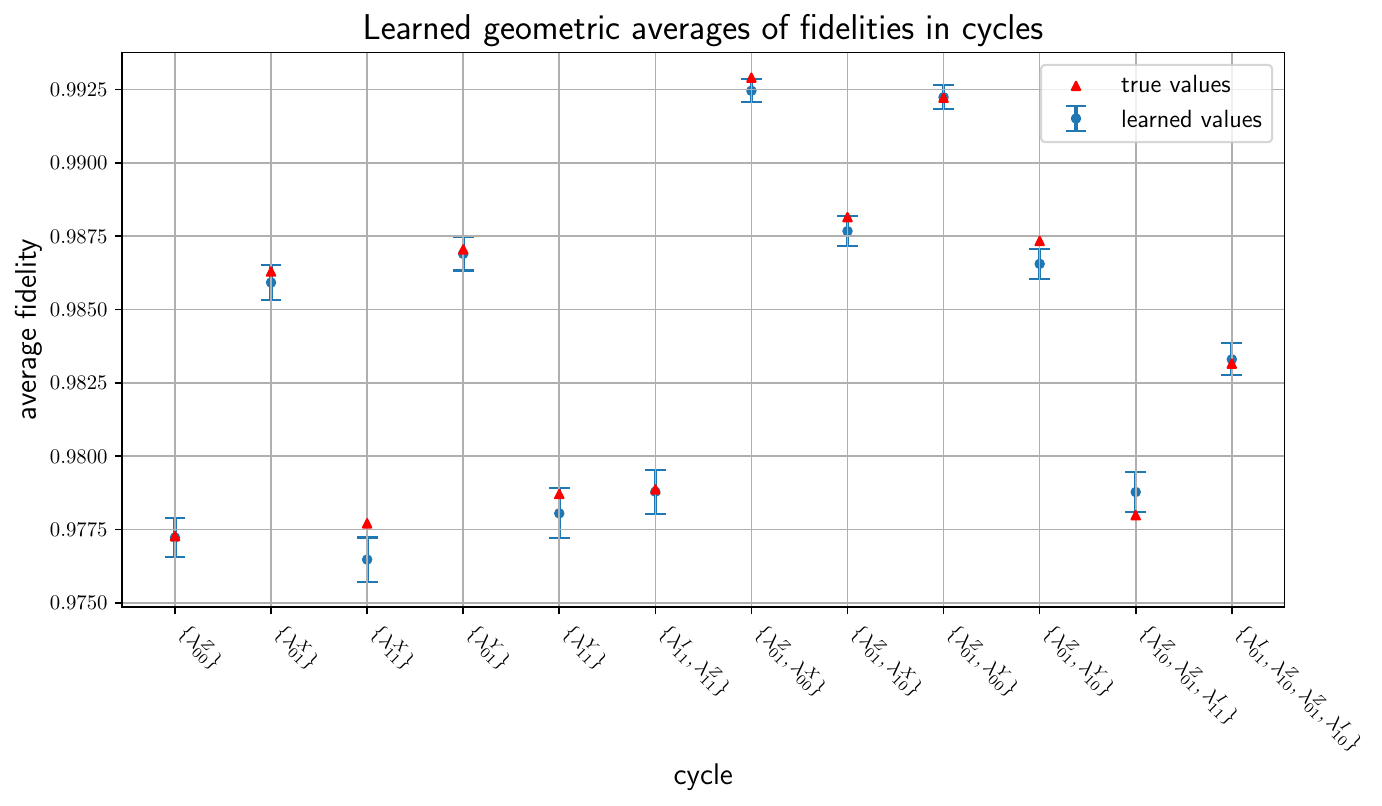"}
\caption{Simulation results in \hyperref[par:fid]{Part a}. There are $13$ learnable degrees of freedom, each corresponds to a cycle. The x-axis shows the corresponding cycles while the y-axis shows the learned geometric average of the fidelities in the cycle. The learnable information $\lambda^I_{00}$ is omitted here since it is always $1$. The error bars represent one standard deviation.}
\label{fig:learnedCycle}
\end{figure}

In the setting above, the variance can also be estimated theoretically.
Suppose $s,t,\lambda$ are all sufficiently close to $1$, the variance is roughly given by\footnote{It is worth noting that a rigorous formula for the variance is unattainable as our protocol has some small probability of failure.
In extreme cases where either $s$ or $t$ in the protocol deviates significantly from their expected value such that $\frac{s}{t}$ becomes negative, the output of the protocol is undefined.
However, assuming that the number of shots taken is sufficient so that such extreme cases can be ignored and that the noise rate is low, the variance of the results can be estimated.}
\begin{align}
\operatorname{Var}\left[e^{\frac{1}{\ell}\log\frac{s}{t}}\right]&\approx\operatorname{Var}\left[\frac{s-t}{\ell}\right]=\frac{\operatorname{Var}(s)+\operatorname{Var}(t)}{\ell^2}=\frac{1-\mathbb{E}[s]^2+1-\mathbb{E}[t]^2}{\ell^2 N/2}\nonumber\\
&=\frac{1-\lambda_0^2\lambda^{2\ell}+1-\lambda_0^2}{\ell^2N/2}\approx\frac{2(1-\lambda_0^2)}{\ell^2N/2}+\frac{\lambda_0^2(1-\lambda^2)}{\ell N/2}.\label{eq:std}
\end{align}
Here $\ell=lL$ is the number of MCMs used in each shot, $N$ is the total number of shots used for estimating $s$ and $t$, and $\lambda_0=\langle\!\langle\mathcal{G}^\dagger(Q_1\otimes Z^{x_1})|\rho\rangle\!\rangle\lambda_M^{v_l}$ is a constant related only to SPAM errors.
We have used the assumption that the noise rate is small, so $s$ and $t$ concentrates around values close to $1$.

\paragraph{Testing independence of measurement and state preparation.}\label{par:test}
We then simulate the experiment of testing the independence of the measurement and state preparation, as described in \cref{sec:application}.
We randomly construct a quantum instrument for a noisy MCM and apply randomized compiling on it.
A noisy MCM with independent Pauli channel noise before and after measurement (a measure and prepare instrument) is used for comparison.
The two MCMs have roughly the same noise level and we simulate our test for them.
Specifically, we learn the correlations $c^Q_{0,1,0,1}$ for $Q\in\{I, X, Y, Z\}$.
Because of the anti-symmetry of the correlations, these are all the correlations that we need to test.
We use an approach slightly different to the previous part to learn the correlations.
For each correlation, we learn the $4$ log fidelities separately using the (modified) protocol and add/subtract them together.
By \cref{lem:indepLearnable} we know that the boundaries will cancel out.
For compiled MCM, $100$ compiled circuits are used for the main circuit, each with $20000$ shots.
$2000000$ shots are used for the auxiliary circuit, so one log fidelity is learned using $4000000$ shots.
The same number of shots are used for measure and prepare instrument.
The standard deviation is estimated via bootstrapping.
Our results are shown in \cref{fig:learnedIndep}.
One can see that our criteria judges the independence nicely.
We note that the true values are always non-negative because of the positivity requirement of the quantum instruments.

\begin{figure}[t]
\centering
\includegraphics[width=0.99\textwidth]{"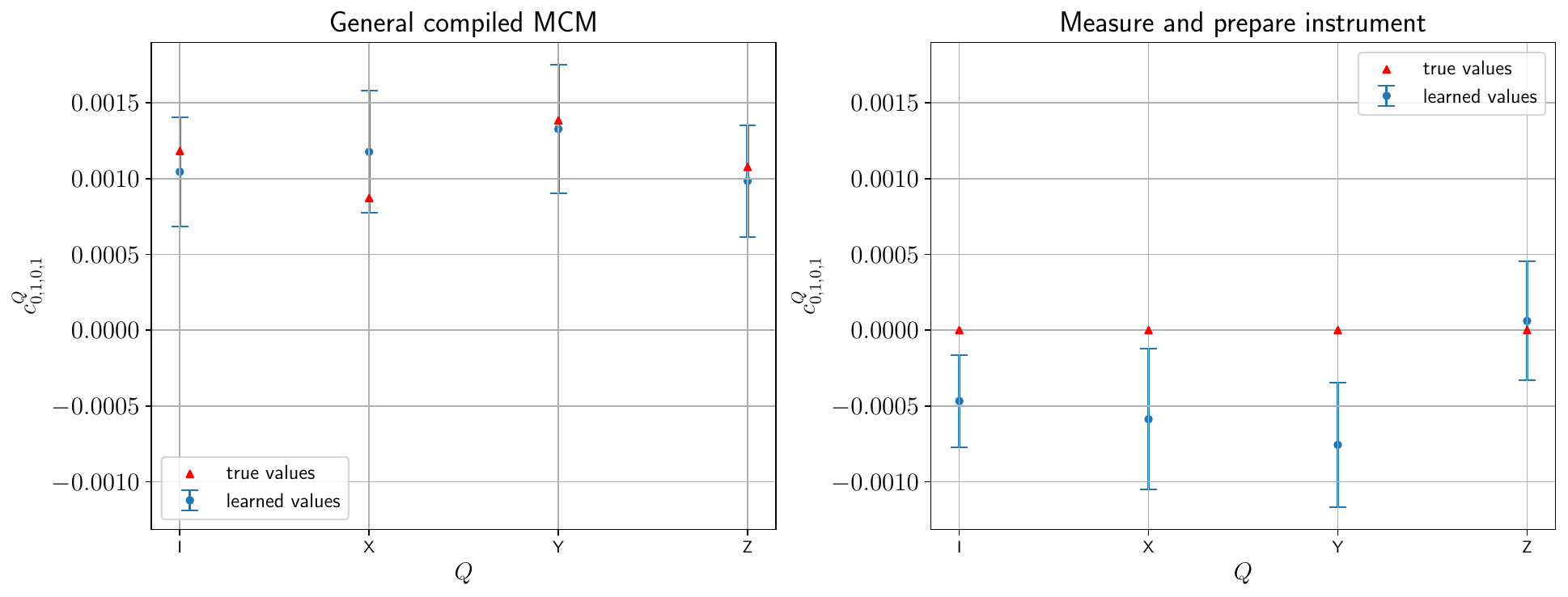"}
\caption{Simulation results in \hyperref[par:test]{Part b}. The left side shows the estimated correlations of a general compiled MCM while the right side shows that of a measure and prepare channel. The $y$-axis represents the learned values of the correlation $c^Q_{0,1,0,1}$. The $x$-axis indicates the corresponding $Q$. The error bars represent one standard deviation. Logarithms are in base $e$.}
\label{fig:learnedIndep}
\end{figure}

\paragraph{Learning Pauli error rate.}
Finally we learn the $p_{1,1}^I$ of a random quantum instrument.
The method has already been used as an example in \cref{eq:decompexp}.
We learn the two terms in \cref{eq:decompexp} by the (modified) protocol, each with $100$ compiled circuits for the main circuit and we take $10$ shots from each of the compiled circuits.
$1000$ shots are taken for the auxiliary circuit.
We repeat the above experiment $10$ times and use their mean as output, so in total $40000$ shots are taken.
The standard deviation is estimated from repetitions.
In our simulation, the true value for $p^I_{1,1}$ is $2.79\times10^{-4}$, while our learned value is $2.60\times10^{-4}$ with estimated standard deviation $0.55\times10^{-4}$.

\section{Discussion and outlook}

In this work, we conduct a comprehensive investigation on learning noisy MCMs. Using graph-theoretic tools, we can determine all the learnable degrees of freedom for a compiled MCM, and design a CB-type protocol to learn all learnable information. We demonstrate our learning protocol in a numerical example. As an application, we illustrate how our protocol can be used to test the independence between measurement and state preparation for an MCM.

Our results put forward a fundamental limit on characterizing randomly compiled MCMs. Namely, we show that certain parameters are coupled with gauge transformations, thus cannot be precisely learned SPAM-robustly. In practice, it could be desirable to obtain knowledge about those ``unlearnable'' degrees of freedom. For example, one might want to know whether the dominant error is in the measurement readout or in the post-measurement state, which might provide insight on how to improve the hardware design. Solving such problem would require additional assumptions to anchor the gauge, e.g., assuming noiseless state preparation or terminating measurements, or use the physicality constraints as investigated in Ref.~\cite{chen2023learnability}.
How to characterize MCMs efficiently under such physical assumptions requires dedicated exploration in the future.


Moreover, the concatenation technique discussed at the end of \cref{sec:protocol} may not always be applicable.
In \cref{sec:dLearn} we proved that all the learnable information are cycles.
For CB with interleaved gates, we can prove that the connected components of the pattern transfer graph are strongly connected.
As a consequence, it permits a directed cycle basis, meaning that all the learnable information can be decomposed into directed cycles and thus concatenated.
However, we cannot obtain similar results for our protocol.
Whether the entire cycle space can be learned through concatenation (i.e. whether the pattern transfer graph admits a circuit basis) remains an open question.

Finally, although we have given a complete characterization of the learnability of noisy MCMs in terms of Pauli fidelities, we are only able to give a partial characterization for the learnability of Pauli error rates with $3$ special cases in \cref{sec:pLearn}. For noisy Clifford gate, in contrast, the learnability of Pauli error rates and Pauli fidelities are basically the same, thanks to the fact that the cycle space is invariant under Walsh-Hadamard transform (see \cref{sec:invar} for details). Such property breaks down when we include MCMs. Therefore, a more comprehensive and physically-motivating understanding of the learnability for Pauli error rates is left for future investigation.

\section*{Code availability}
The code for numerical simulation can be found at \url{https://github.com/zhihan-z/MCM_Learnability}.

\begin{acknowledgments}
We thank Dongling Deng, Dong Yuan, Weiyuan Gong, Qi Ye, and Alireza Seif for helpful discussions.
S.C. and L.J. acknowledge support from the ARO (W911NF-23-1-0077), ARO MURI (W911NF-21-1-0325), AFOSR MURI (FA9550-19-1-0399, FA9550-21-1-0209, FA9550-23-1-0338), DARPA (HR0011-24-9-0359, HR0011-24-9-0361), NSF (OMA-1936118, ERC-1941583, OMA-2137642, OSI-2326767, CCF-2312755), NTT Research, Samsung GRO, Packard Foundation (2020-71479). Y.L.~is supported by DOE Grant No. DE-SC0024124, NSF Grant No. 2311733, and DOE Quantum Systems Accelerator. 

Note added -- Shortly after we posted our manuscript, an independent work proposed a similar algorithm~\cite{hines2024pauli}.
\end{acknowledgments}

\bibliography{ref}

\appendix
\section{Randomized compiling}
\subsection{Randomized compiling for MCMs}
\label{app:MCMrc}
In this section, we explain how randomized compiling techniques can reduce a general quantum instrument $\mathcal{M}=\{\mathcal{M}_k\}$ into an instrument $\mathcal{T}=\{\mathcal{T}_k\}$ that fits our model.
The scheme we adopt is a simple variation of the scheme proposed in~\cite{beale2023randomized}.
Fig.~\ref{fig:comp} illustrates our randomized compiling scheme.

\begin{figure}[h]
    \centering
    \begin{tabular}{c}
\Qcircuit @C=1em @R=.7em {
    & \push{\rule{3em}{0em}} & & /^{m} \qw & \gate{P_1} & \multimeasure{1}{\mathcal{M}} & \gate{P} & \qw & \qw & & & & /^{m} \qw & \multimeasure{1}{\mathcal{T}} & \qw\\
    & \mathbb{E}_{P,x,y,z} & & /^{n} \qw & \gate{P_2} & \ghost{\mathcal{M}} & \gate{Z^\gamma} & \gate{X^\alpha} & \qw & & = & & /^{n} \qw & \ghost{\mathcal{T}} & \qw\\
    & \push{\rule{3em}{0em}} & & & & \control \cwx &\cgate{add~\alpha} & \cw & \cw & \lket{k} & & & & \control \cwx & \cw & \lket{k}\\
    & & & & & k+\alpha\\
}
\end{tabular}
\caption{Illustration of the randomized compiling scheme.}
\label{fig:comp}
\end{figure}

Specifically, we choose $P\in\mathcal{P}^m$, $\alpha,\beta,\gamma\in\mathbb{Z}_2^n$ uniformly randomly and define operator $P_1\otimes P_2=\mathcal{G}^\dagger(P\otimes Z^\beta X^\alpha)$.
We apply gate $P_1\otimes P_2$, followed by the noisy MCM $\mathcal{M}$, and then gate $P\otimes X^\alpha Z^\gamma$.
The classical output from the MCM is added by $\alpha$ to form the measurement result of the compiled channel.
We make the assumption that all single qubit unitaries can be implemented noiselessly.
This is justified by the fact that noises of multi-qubit operations are usually much higher than that of single-qubit operations.
Then, after the compiling, the effective instrument becomes
\begin{align}
\mathcal{T}_k&=\operatornamewithlimits{\mathbb{E}}_{P,\alpha,\beta,\gamma}\left[(\mathcal{P}\otimes \mathcal{X}^\alpha\mathcal{Z}^\gamma)\mathcal{M}_{k-\alpha}(\mathcal{G}^\dagger(\mathcal{P}\otimes \mathcal{Z}^\beta\mathcal{X}^\alpha)\mathcal{G})\right]\\
&=\operatornamewithlimits{\mathbb{E}}_{P,\alpha,\beta,\gamma}\left[(\mathcal{P}\otimes \mathcal{X}^\alpha\mathcal{Z}^\gamma)\mathcal{M}_{k-\alpha}\mathcal{G}^\dagger(\mathcal{P}\otimes \mathcal{Z}^\beta\mathcal{X}^\alpha)\right]\mathcal{G}.
\end{align}
Here, gates and channels are used interchangeably.
To simplify, we need a lemma from Ref.~\cite{beale2023randomized}.
\begin{lemma}[Theorem 2, Ref.~\cite{beale2023randomized}]\label{lem:comp}
The following randomized compiling twirls a general quantum instrument $\mathcal{M}=\{\mathcal{M}_k\}$ to a uniform stochastic instrument $\{\mathcal{U}_k\}$:
\begin{equation}
\operatornamewithlimits{\mathbb{E}}_{P,\alpha,\beta,\gamma}\left[(\mathcal{P}\otimes \mathcal{X}^\alpha\mathcal{Z}^\gamma)\mathcal{M}_{k-\alpha}\mathcal{G}^\dagger(\mathcal{P}\otimes \mathcal{Z}^\beta\mathcal{X}^\alpha)\right]=\mathcal{U}_k.
\end{equation}
\end{lemma}
Intuitively, the random $\alpha$ averages over different absolute measurement outcomes, making the $\Lambda_{a,b}$ in $\mathcal{U}_k$ (\cref{eq:USIc}) independent of $k$ and depends instead only on the relative differences $a$ and $b$.

From \cref{lem:comp}, we immediately see that $\mathcal{T}_k=\mathcal{U}_k\mathcal{G}$ and thus justifies our model.
We note that this compiling scheme is slightly different from Figure 5 of Ref.~\cite{beale2023randomized} as we are not allowed to insert twirling gates between the Clifford gate $G$ and the subsystem measurement since we assume that they are implemented as a whole.

\subsection{Randomized compiling for terminating measurements}\label{app:TMrc}
We model a terminating measurement as a (special) Pauli channel $\Lambda_M$, followed by an ideal terminating measurement.
The Pauli channel here has the form
\begin{equation}
\Lambda_M=\sum_{F\in\mathcal{P}^{n+m}}\lambda_M^{pt(F)}|F\rangle\!\rangle\langle\!\langle F|.
\end{equation}
One can immediately see that this Pauli channel is special in the sense that the Pauli fidelities are indexed by Pauli weight patterns instead of Pauli operators.
We note that this is for mathematical convenience, and our proof also work for general Pauli channels.

Again, our noise model can be realized via randomized compiling.
Assume the noisy terminating measurements implemented can be viewed as an arbitrary channel $\mathcal{E}^M$ followed by an ideal terminating measurement.
Then we can use the standard randomized compiling~\cite{wallman2016noise} techniques to twirl the noise into a special Pauli channel $\Lambda^M$.
More specifically, suppose we want to measure observable $E$, then we shall select a random $H\in\mathcal{C}^{\otimes n+m}$, apply gate $H$ and then measure in $\mathcal{H}(E)$ basis.
\begin{lemma}\label{lem:term}
If we use the above compiled terminating measurements to measure observable $E\in\mathcal{P}^{n+m}$ on $\rho$, then the expectation of the result is $\lambda_M^{pt(E)}\langle\!\langle E|\rho\rangle\!\rangle$.
Here $\lambda_M$ are actually Pauli fidelities of the terminating measurement noise.
It is indexed by Pauli weight patterns.
\end{lemma}
\begin{proof}
The expectation value is
\begin{align}
\operatornamewithlimits{\mathbb{E}}_{H\in\mathcal{C}^{\otimes n+m}}\left[\Tr(\mathcal{H}(E)\cdot\mathcal{E}\circ \mathcal{H}(\rho))\right]=tr\left(E\operatornamewithlimits{\mathbb{E}}_{H\in\mathcal{C}^{\otimes n+m}}\left[\mathcal{H}^\dagger\circ\mathcal{E}\circ \mathcal{H}(\rho)\right]\right)=\lambda_M^{pt(E)}\langle\!\langle E|\rho\rangle\!\rangle.
\end{align}
The second equality is the direct consequence of Schur's lemma~\cite{fulton2013representation} as in standard randomized compiling.
\end{proof}

\section{Structure of the measure-and-prepare instruments}\label{sec:refresh}
In this section, we formalize the definition of measure-and-prepare instrument and express it in the dual space.
The process is shown in \cref{fig:MaPIllu} and follows the description in \cref{sec:USIRIntro}.

We start with a uniform stochastic instrument (we use the dual space form \cref{eq:fidelityform} here).
Without the Clifford gate $G$, a uniform stochastic instrument is
\begin{equation}
\mathcal{U}_k=\frac{1}{2^{2n+m}}\sum_{x,y\in\mathbb{Z}_2^n,Q\in\mathcal{P}^m}(-1)^{k\cdot(x+y)}\lambda_{x,y}^Q|Q\otimes Z^y\rangle\!\rangle\langle\!\langle Q\otimes Z^x|.
\end{equation}
Suppose we obtain the measurement outcome $k$ which corresponds to $\mc U_k$ being applied. A measure-and-prepare instrument corresponds to tracing out the ancilla register, followed by a preparation of fresh ancilla with $X$ correction (\cref{fig:MaPIllu}). Our goal is to derive the dual-space form of this process.

First, tracing out the ancilla register in $\mc U_k$, we get
\begin{equation}
\frac{1}{2^{n+m}}\sum_{x\in\mathbb{Z}_2^n,Q\in\mathcal{P}^m}(-1)^{k\cdot x}\lambda_{x,0}^Q|Q\rangle\!\rangle\langle\!\langle Q\otimes Z^x|=\frac{1}{2^{n+m}}\sum_{x\in\mathbb{Z}_2^n,Q\in\mathcal{P}^m}(-1)^{k\cdot x}\zeta_x^Q|Q\rangle\!\rangle\langle\!\langle Q\otimes Z^x|.
\end{equation}
Here we have defined $\zeta_x^Q=\lambda_{x,0}^Q$. Next, we model the state-preparation procedure as preparing a perfect initial state $\lket{0^n}$ (the computational basis state where all qubits are zero) followed by a Pauli noise channel, so the next step is to put the refreshed ancillas in, 
\begin{equation}
\frac{1}{2^{2n+m}}\sum_{x,y\in\mathbb{Z}_2^n,Q\in\mathcal{P}^m}(-1)^{k\cdot x}\zeta_x^Q|Q\otimes Z^y\rangle\!\rangle\langle\!\langle Q\otimes Z^x|.
\end{equation}
For the Pauli noise channel, denote its Pauli fidelity for $Q\otimes Z^y$ by $\xi_y^Q$.
Applying this noise, we get
\begin{equation}
\frac{1}{2^{2n+m}}\sum_{x,y\in\mathbb{Z}_2^n,Q\in\mathcal{P}^m}(-1)^{k\cdot x}\zeta_x^Q\xi_y^Q|Q\otimes Z^y\rangle\!\rangle\langle\!\langle Q\otimes Z^x|.
\end{equation}
Apply the NOT gates $X^k$, we get the effective channel for measurement result $k$
\begin{equation}\label{eq:USIRd}
\mathcal{U}'_k=\frac{1}{2^{2n+m}}\sum_{x,y\in\mathbb{Z}_2^n,Q\in\mathcal{P}^m}(-1)^{k\cdot (x+y)}\zeta_x^Q\xi_y^Q|Q\otimes Z^y\rangle\!\rangle\langle\!\langle Q\otimes Z^x|.
\end{equation}
$\{\mathcal{U}'_k\}$ is the measure-and-prepare instrument expressed in dual space.
Compared with \cref{eq:USId}, one can see that $\{\mathcal{U}'_k\}$ is also a uniform stochastic instrument, with the additional structure
\begin{equation}\label{eq:USIR_fid}
\lambda'^Q_{x,y}=\zeta_x^Q\xi_y^Q,
\end{equation}
i.e., the Pauli fidelity is factorized into a product.

In terms of error rates, define $q$ and $r$ as the inverse Fourier transformation of $\zeta$ and $\xi$ over $\mathbb{Z}_2^n$ respectively,
\begin{equation}\label{eq:USIRd2p}
q_a^P=\frac{1}{2^{n+2m}}\sum_{x\in\mathbb{Z}_2^n,Q\in\mathcal{P}^m}(-1)^{a\cdot x+\langle P,Q\rangle}\zeta_x^Q,\quad r_b^P=\frac{1}{2^{n+2m}}\sum_{x\in\mathbb{Z}_2^n,Q\in\mathcal{P}^m}(-1)^{b\cdot x+\langle P,Q\rangle}\xi^Q_x.
\end{equation}
$q$ and $r$ corresponds to the Pauli error rates for the measurement and prepare process, respectively.
The inverse transformations are
\begin{equation}\label{eq:USIRp2d}
\zeta_x^Q=\sum_{a\in\mathbb{Z}_2^n,P\in\mathcal{P}^m}(-1)^{a\cdot x+\langle P,Q\rangle}q_a^P,\quad \xi^Q_x=\sum_{a\in\mathbb{Z}_2^n,P\in\mathcal{P}^m}(-1)^{b\cdot x+\langle P,Q\rangle}r_b^P.
\end{equation}
And we have
\begin{align}
p'^P_{a,b}&=\frac{1}{4^{n+m}}\sum_{x,y\in\mathbb{Z}_2^n,Q\in\mathcal{P}^m}(-1)^{a\cdot x+b\cdot y+\langle P,Q\rangle}\lambda'^Q_{x,y}\\
&=\frac{1}{4^{n+m}}\sum_{x,y\in\mathbb{Z}_2^n,Q\in\mathcal{P}^m}(-1)^{a\cdot x+b\cdot y+\langle P,Q\rangle}\zeta_x^Q\xi_y^Q\\
&=\frac{1}{4^{n+2m}}\sum_{x,y\in\mathbb{Z}_2^n,Q_1,Q_2,P\in\mathcal{P}^m}(-1)^{a\cdot x+b\cdot y+\langle P,Q_1Q_2\rangle+\langle P,Q_2\rangle}\zeta_x^{Q_1}\xi_y^{Q_2}\label{eq:refreshdelta}\\
&=\frac{1}{4^{n+2m}}\sum_{x,y\in\mathbb{Z}_2^n,Q_1,Q_2,P\in\mathcal{P}^m}(-1)^{a\cdot x+b\cdot y+\langle P,Q_1\rangle+\langle PP_1,Q_2\rangle}\zeta_x^{Q_1}\xi_y^{Q_2}\\
&=\sum_{P_1\cdot P_2=P}\left(\frac{1}{2^{n+2m}}\sum_{x\in\mathbb{Z}_2^n,Q\in\mathcal{P}^m}(-1)^{a\cdot x+\langle P,Q\rangle}\zeta_x^Q\right)\left(\frac{1}{2^{n+2m}}\sum_{x\in\mathbb{Z}_2^n,Q\in\mathcal{P}^m}(-1)^{b\cdot x+\langle P,Q\rangle}\xi^Q\right)\\
&=\sum_{P_1\cdot P_2=P}q_a^{P_1}r_b^{P_2}.\label{eq:USIR_erate}
\end{align}
\cref{eq:refreshdelta} uses the fact that $\sum_{P\in\mathcal{P}^m}(-1)^{\langle P,Q_1Q_2\rangle}$ is non-zero only when $Q_1=Q_2$ and takes value $4^m$ in this case.
By plugging in this equation to \cref{eq:USIp} we get the physical space expression of the measure-and-prepare instrument.
\begin{equation}\label{eq:USIRp}
\mathcal{U}_k'=\sum_{a,b\in\mathbb{Z}_2^n,P\in\mathcal{P}^m,P_1\cdot P_2=P}q_a^{P_1}r_b^{P_2}(\mathcal{P}\otimes\lket{k+b}\lbra{k+a}).
\end{equation}

\section{Proofs from~\cref{sec:pLearn}}\label{sec:defer}
\noindent\paragraph*{Proof of~\cref{prop:1}}
It suffice to show the corresponding $1$-chain is in $Z$.
\begin{align}
&\partial\frac{1}{4^{n+m}}\sum_{x,y\in\mathbb{Z}_2^n,Q\in\mathcal{P}^m}(-1)^{a\cdot x+b\cdot y+\langle P,Q\rangle}e^Q_{x,y}\\
=&\frac{1}{4^{n+m}}\sum_{x,y\in\mathbb{Z}_2^n,Q\in\mathcal{P}^m}(-1)^{a\cdot x+b\cdot y+\langle P,Q\rangle}(-pt(\mathcal{G}^\dagger(Q\otimes Z^x))+pt(Q\otimes Z^y))\label{eq:claim1boundary}\\
=&-\frac{1}{4^{n+m}}\sum_{x\in\mathbb{Z}_2^n,Q\in\mathcal{P}^m}(-1)^{a\cdot x+\langle P,Q\rangle}pt(\mathcal{G}^\dagger(Q\otimes Z^x))\sum_{y\in\mathbb{Z}_2^n}(-1)^{b\cdot y}\nonumber\\
&+\frac{1}{4^{n+m}}\sum_{y\in\mathbb{Z}_2^n,Q\in\mathcal{P}^m}(-1)^{b\cdot y+\langle P,Q\rangle}pt(Q\otimes Z^y)\sum_{x\in\mathbb{Z}_2^n}(-1)^{a\cdot x}=0.\label{eq:claim1sum0}
\end{align}
Note that the addition of Pauli weight patterns are performed in the sense of $0$-chains.
\cref{eq:claim1boundary} uses the definition of boundary operator and \cref{eq:claim1sum0} uses the fact that when $a\neq0$ and $b\neq0$, $\sum_{y\in\mathbb{Z}_2^n}(-1)^{b\cdot y}=\sum_{x\in\mathbb{Z}_2^n}(-1)^{a\cdot x}=0$.

\noindent\paragraph*{Proof of~\cref{prop:2}}
\begin{align}
&\partial\frac{1}{4^{n+m}}\sum_{x,y,k\in\mathbb{Z}_2^n,Q\in\mathcal{P}^m}(-1)^{\langle S^kP,Q\rangle}e^Q_{x,y}\\
=&\frac{1}{4^{n+m}}\sum_{x,y\in\mathbb{Z}_2^n,Q\in\mathcal{P}^m}(-1)^{\langle P,Q\rangle}\partial e^Q_{x,y}\sum_{k\in\mathbb{Z}_2^n}(-1)^{\langle S^k,Q\rangle}\\
=&\frac{2^n}{4^{n+m}}\sum_{x,y\in\mathbb{Z}_2^n,Q\in C_{\mathcal{P}^m}(\mathcal{S})}(-1)^{\langle P,Q\rangle}(-pt(\mc G^\dagger (Q\otimes Z^x))+pt(Q\otimes Z^y))\label{eq:claim2centralizersum}\\
=&\frac{2^n}{4^{n+m}}\sum_{x,y\in\mathbb{Z}_2^n,Q\in C_{\mathcal{P}^m}(\mathcal{S})}(-1)^{\langle P,Q\rangle}(-pt(S^xQ\otimes Z^x)+pt(Q\otimes Z^y))\\
=&-\frac{1}{4^m}\sum_{x\in\mathbb{Z}_2^n,Q\in C_{\mathcal{P}^m}(\mathcal{S})}(-1)^{\langle P,Q\rangle}pt(S^xQ\otimes Z^x)+\frac{1}{4^m}\sum_{y\in\mathbb{Z}_2^n,Q\in C_{\mathcal{P}^m}(\mathcal{S})}(-1)^{\langle P,Q\rangle}pt(Q\otimes Z^y)\\
=&-\frac{1}{4^m}\sum_{x\in\mathbb{Z}_2^n,Q\in C_{\mathcal{P}^m}(\mathcal{S})}(-1)^{\langle P,S^xQ\rangle}pt(Q\otimes Z^x)+\frac{1}{4^m}\sum_{x\in\mathbb{Z}_2^n,Q\in C_{\mathcal{P}^m}(\mathcal{S})}(-1)^{\langle P,Q\rangle}pt(Q\otimes Z^x)=0.\label{eq:claim2sum0}
\end{align}
\cref{eq:claim2centralizersum} uses the fact that $\sum_{k\in\mathbb{Z}_2^n}(-1)^{\langle S^k,Q\rangle}$ is non-zero and takes value $2^n$ only when $Q\in C_{\mathcal{P}^m}(\mathcal{S})$ and \cref{eq:claim2sum0} makes a substitution $Q\mapsto S^xQ$ for the first term and also uses the fact $P\in C_{\mathcal{P}^m}(\mathcal{S})$.

\noindent\paragraph*{Proof of~\cref{prop:3}}
\begin{align}
&\partial\frac{1}{4^{n+m}}\sum_{Q\in\mathcal{P}^m,x,y,k\in\mathbb{Z}_2^n}(-1)^{\langle S^kP,Q\rangle}\left[(-1)^{a\cdot x}+(-1)^{a\cdot y}\right]e^Q_{x,y}\\
=&\frac{1}{4^{n+m}}\sum_{Q\in\mathcal{P}^m,x,y\in\mathbb{Z}_2^n}(-1)^{\langle P,Q\rangle}\left[(-1)^{a\cdot x}+(-1)^{a\cdot y}\right]\partial e^Q_{x,y}\sum_{k\in\mathbb{Z}_2^n}(-1)^{\langle S^k,Q\rangle}\\
=&\frac{2^n}{4^{n+m}}\sum_{Q\in C_{\mathcal{P}^m}(\mathcal{S}),x,y\in\mathbb{Z}_2^n}(-1)^{\langle P,Q\rangle}\left[(-1)^{a\cdot x}+(-1)^{a\cdot y}\right]\partial e^Q_{x,y}\label{eq:claim3centralizersum}\\
=&-\frac{2^n}{4^{n+m}}\sum_{Q\in C_{\mathcal{P}^m}(\mathcal{S}),x,y\in\mathbb{Z}_2^n}(-1)^{\langle P,Q\rangle}\left[(-1)^{a\cdot x}+(-1)^{a\cdot y}\right]pt(S^xQ\otimes Z^x)\nonumber\\
&+\frac{2^n}{4^{n+m}}\sum_{Q\in C_{\mathcal{P}^m}(\mathcal{S}),x,y\in\mathbb{Z}_2^n}(-1)^{\langle P,Q\rangle}\left[(-1)^{a\cdot x}+(-1)^{a\cdot y}\right]pt(Q\otimes Z^y)\\
=&-\frac{1}{4^m}\sum_{Q\in C_{\mathcal{P}^m}(\mathcal{S}),x\in\mathbb{Z}_2^n}(-1)^{a\cdot x+\langle P,Q\rangle}pt(S^xQ\otimes Z^x)+\frac{1}{4^m}\sum_{Q\in C_{\mathcal{P}^m}(\mathcal{S}),y\in\mathbb{Z}_2^n}(-1)^{a\cdot y+\langle P,Q\rangle}pt(Q\otimes Z^y)\\
=&-\frac{1}{4^m}\sum_{Q\in C_{\mathcal{P}^m}(\mathcal{S}),x\in\mathbb{Z}_2^n}(-1)^{a\cdot x+\langle P,S^xQ\rangle}pt(Q\otimes Z^x)+\frac{1}{4^m}\sum_{Q\in C_{\mathcal{P}^m}(\mathcal{S}),x\in\mathbb{Z}_2^n}(-1)^{a\cdot x+\langle P,Q\rangle}pt(Q\otimes Z^x)=0.\label{eq:claim3sum0}
\end{align}
\cref{eq:claim3centralizersum} uses the trick similar to \cref{eq:claim2centralizersum} and \cref{eq:claim3sum0} uses trick similar to \cref{eq:claim2sum0}.

\section{Invariance of cycle space under Walsh-Hadamard transformation}\label{sec:invar}
In \cite{chen2023learnability}, the authors made a conjecture about the learnability of Pauli error rates for noisy Clifford gates.
They conjectured that in this case, the cycle space is invariant under Walsh-Hadamard transformation.
That is, for a learnable cycle, if we substitute the log fidelities by the corresponding Pauli error rates, the resulting quantity about Pauli error rates is also learnable (under first order approximation).
Though noisy Clifford gates is not the focus of this paper and from \cref{sec:pLearn} obviously this special property no longer holds in our generalized case, to complete the physical vs dual space picture, in this section we give a proof for this conjecture.

We note by the way that though in a different language, Lemma 3 in Ref.~\cite{carignan2023error} effectively gives a constructive proof for the theorem for the special case where no interleaving single qubit Clifford gates are allowed.
But in this section, we follow the non-constructive proof fashion in \cref{sec:pLearn} to prove the general theorem.

Once again, we emphasize that in this section we follow our convention and model the noise as happening after the Clifford gate, so our definition of pattern transfer graph is different from that of Ref.~\cite{chen2023learnability}.
Ref.~\cite{chen2023learnability} has proved that for Pauli fidelities, the cycle space is learnable while the cut space is unlearnable, and their results are still valid using our convention.

\begin{theorem}\label{thm:invar}
For a cycle $v_0,e^{Q_1},v_1,\ldots,e^{Q_l},v_0$, $\sum_{k=1}^lp^{Q_k}$ is learnable.
\end{theorem}

\begin{proof}
Since we are considering a cycle, view the subscript $0$ equivalent as the subscript $l$.
By definition of pattern transfer graph, $\exists H_1,\ldots,H_l\in\mathcal{C}^{\otimes m}$ ($\mathcal{C}$ is the group of single qubit Clifford gates) s.t. $H_{k-1}Q_{k-1}H_{k-1}^\dagger=G^\dagger Q_kG$.
Then it suffices to prove the learnability of the corresponding $1$-chain.

\begin{align}
&\partial\frac{1}{4^m}\sum_{k=1}^l\sum_{R\in\mathcal{P}^m}(-1)^{\langle Q_k,R\rangle}e^R\\
=&\frac{1}{4^m}\sum_{R\in\mathcal{P}^m}\partial e^R\sum_{k=1}^l(-1)^{\langle Q_k,R\rangle}\\
=&\frac{1}{4^m}\sum_{R\in\mathcal{P}^m}pt(R)\sum_{k=1}^l(-1)^{\langle Q_k,R\rangle}-\frac{1}{4^m}\sum_{R\in\mathcal{P}^m}pt(\mathcal{G}^\dagger(R))\sum_{k=1}^l(-1)^{\langle Q_k,R\rangle}\label{eq:SuppPartial}\\
=&\frac{1}{4^m}\sum_{R\in\mathcal{P}^m}pt(R)\sum_{k=1}^l(-1)^{\langle Q_k,R\rangle}-\frac{1}{4^m}\sum_{R\in\mathcal{P}^m}pt(R)\sum_{k=1}^l(-1)^{\langle GH_{k-1}Q_{k-1}H_{k-1}^\dagger G^\dagger,GRG^\dagger\rangle}\label{eq:SuppSub1}\\
=&\frac{1}{4^m}\sum_{R\in\mathcal{P}^m}pt(R)\sum_{k=1}^l(-1)^{\langle Q_k,R\rangle}-\frac{1}{4^m}\sum_{R\in\mathcal{P}^m}pt(R)\sum_{k=1}^l(-1)^{\langle Q_{k-1},H_{k-1}^\dagger RH_{k-1}\rangle}\label{eq:SuppComm}\\
=&\frac{1}{4^m}\sum_{R\in\mathcal{P}^m}pt(R)\sum_{k=1}^l(-1)^{\langle Q_k,R\rangle}-\frac{1}{4^m}\sum_{R\in\mathcal{P}^m}pt(R)\sum_{k=1}^l(-1)^{\langle Q_{k-1},R\rangle}\label{eq:SuppSub2}\\
=&0
\end{align}
\cref{eq:SuppPartial} uses the definition of pattern transfer graph for cycle benchmarking.
\cref{eq:SuppSub1} makes the substitution $R\mapsto GRG^\dagger$.
\cref{eq:SuppComm} uses the fact that conjugating by the same Clifford gate does not change the commutation relation.
\cref{eq:SuppSub2} makes the substitution $R\mapsto H_{k-1}RH_{k-1}^\dagger$ and uses the fact that single qubit Clifford gates do not change patterns.
\end{proof}
By simple counting of degrees of freedom, \cite{chen2023learnability} proved that \cref{thm:invar} implies that the learnability of Pauli error rates are also completely characterized by cycle space.
That is, the cycle space is the space of all learnable about Pauli error rates.

\end{document}

%% file: bigTikz.tex
\begin{figure}[t]
\centering
\subfloat[Composition of noisy Clifford gates($G=\mathrm{CNOT}$, $m=2$, $n=0$, with no interleaving single qubit Clifford gates) in the dual space. The noisy Clifford gates take one Pauli operator to one Pauli operator, so different paths can not have the same start point and end point in the same time (highlighted by blue dash-dot lines), and thus they do not superpose with each other.\label{fig:PCcompose}]{\centering

\begin{tikzpicture}
[pt/.style={circle,draw,thick,inner sep=0pt,minimum size=0.5em, scale=1}]
\path[->] (-20pt, 0pt) edge[ultra thick, -{Stealth[length=10pt]}] node[sloped, anchor=center, above] {Time} (-20pt, 160pt);
\node[pt, label=below:{\small II}] ("[0 0 0 0] 0") at (0pt,0pt) {};
\node[pt] ("[0 0 0 0] 1") at (0pt,40pt) {};
\node[pt] ("[0 0 0 0] 2") at (0pt,80pt) {};
\node[pt] ("[0 0 0 0] 3") at (0pt,120pt) {};
\node[pt] ("[0 0 0 0] 4") at (0pt,160pt) {};
\node[pt, label=below:{\small IZ}] ("[0 0 0 1] 0") at (40pt,0pt) {};
\node[pt] ("[0 0 0 1] 1") at (40pt,40pt) {};
\node[pt] ("[0 0 0 1] 2") at (40pt,80pt) {};
\node[pt] ("[0 0 0 1] 3") at (40pt,120pt) {};
\node[pt] ("[0 0 0 1] 4") at (40pt,160pt) {};
\node[pt, label=below:{\small ZI}] ("[0 0 1 0] 0") at (60pt,0pt) {};
\node[pt] ("[0 0 1 0] 1") at (60pt,40pt) {};
\node[pt] ("[0 0 1 0] 2") at (60pt,80pt) {};
\node[pt] ("[0 0 1 0] 3") at (60pt,120pt) {};
\node[pt] ("[0 0 1 0] 4") at (60pt,160pt) {};
\node[pt, label=below:{\small ZZ}] ("[0 0 1 1] 0") at (20pt,0pt) {};
\node[pt] ("[0 0 1 1] 1") at (20pt,40pt) {};
\node[pt] ("[0 0 1 1] 2") at (20pt,80pt) {};
\node[pt] ("[0 0 1 1] 3") at (20pt,120pt) {};
\node[pt] ("[0 0 1 1] 4") at (20pt,160pt) {};
\node[pt, label=below:{\small IX}] ("[0 1 0 0] 0") at (80pt,0pt) {};
\node[pt] ("[0 1 0 0] 1") at (80pt,40pt) {};
\node[pt] ("[0 1 0 0] 2") at (80pt,80pt) {};
\node[pt] ("[0 1 0 0] 3") at (80pt,120pt) {};
\node[pt] ("[0 1 0 0] 4") at (80pt,160pt) {};
\node[pt, label=below:{\small IY}] ("[0 1 0 1] 0") at (100pt,0pt) {};
\node[pt] ("[0 1 0 1] 1") at (100pt,40pt) {};
\node[pt] ("[0 1 0 1] 2") at (100pt,80pt) {};
\node[pt] ("[0 1 0 1] 3") at (100pt,120pt) {};
\node[pt] ("[0 1 0 1] 4") at (100pt,160pt) {};
\node[pt, label=below:{\small ZX}] ("[0 1 1 0] 0") at (140pt,0pt) {};
\node[pt] ("[0 1 1 0] 1") at (140pt,40pt) {};
\node[pt] ("[0 1 1 0] 2") at (140pt,80pt) {};
\node[pt] ("[0 1 1 0] 3") at (140pt,120pt) {};
\node[pt] ("[0 1 1 0] 4") at (140pt,160pt) {};
\node[pt, label=below:{\small ZY}] ("[0 1 1 1] 0") at (120pt,0pt) {};
\node[pt] ("[0 1 1 1] 1") at (120pt,40pt) {};
\node[pt] ("[0 1 1 1] 2") at (120pt,80pt) {};
\node[pt] ("[0 1 1 1] 3") at (120pt,120pt) {};
\node[pt] ("[0 1 1 1] 4") at (120pt,160pt) {};
\node[pt, label=below:{\small XI}] ("[1 0 0 0] 0") at (160pt,0pt) {};
\node[pt] ("[1 0 0 0] 1") at (160pt,40pt) {};
\node[pt] ("[1 0 0 0] 2") at (160pt,80pt) {};
\node[pt] ("[1 0 0 0] 3") at (160pt,120pt) {};
\node[pt] ("[1 0 0 0] 4") at (160pt,160pt) {};
\node[pt, label=below:{\small XZ}] ("[1 0 0 1] 0") at (200pt,0pt) {};
\node[pt] ("[1 0 0 1] 1") at (200pt,40pt) {};
\node[pt] ("[1 0 0 1] 2") at (200pt,80pt) {};
\node[pt] ("[1 0 0 1] 3") at (200pt,120pt) {};
\node[pt] ("[1 0 0 1] 4") at (200pt,160pt) {};
\node[pt, label=below:{\small YI}] ("[1 0 1 0] 0") at (240pt,0pt) {};
\node[pt] ("[1 0 1 0] 1") at (240pt,40pt) {};
\node[pt] ("[1 0 1 0] 2") at (240pt,80pt) {};
\node[pt] ("[1 0 1 0] 3") at (240pt,120pt) {};
\node[pt] ("[1 0 1 0] 4") at (240pt,160pt) {};
\node[pt, label=below:{\small YZ}] ("[1 0 1 1] 0") at (280pt,0pt) {};
\node[pt] ("[1 0 1 1] 1") at (280pt,40pt) {};
\node[pt] ("[1 0 1 1] 2") at (280pt,80pt) {};
\node[pt] ("[1 0 1 1] 3") at (280pt,120pt) {};
\node[pt] ("[1 0 1 1] 4") at (280pt,160pt) {};
\node[pt, label=below:{\small XX}] ("[1 1 0 0] 0") at (180pt,0pt) {};
\node[pt] ("[1 1 0 0] 1") at (180pt,40pt) {};
\node[pt] ("[1 1 0 0] 2") at (180pt,80pt) {};
\node[pt] ("[1 1 0 0] 3") at (180pt,120pt) {};
\node[pt] ("[1 1 0 0] 4") at (180pt,160pt) {};
\node[pt, label=below:{\small XY}] ("[1 1 0 1] 0") at (260pt,0pt) {};
\node[pt] ("[1 1 0 1] 1") at (260pt,40pt) {};
\node[pt] ("[1 1 0 1] 2") at (260pt,80pt) {};
\node[pt] ("[1 1 0 1] 3") at (260pt,120pt) {};
\node[pt] ("[1 1 0 1] 4") at (260pt,160pt) {};
\node[pt, label=below:{\small YX}] ("[1 1 1 0] 0") at (300pt,0pt) {};
\node[pt] ("[1 1 1 0] 1") at (300pt,40pt) {};
\node[pt] ("[1 1 1 0] 2") at (300pt,80pt) {};
\node[pt] ("[1 1 1 0] 3") at (300pt,120pt) {};
\node[pt] ("[1 1 1 0] 4") at (300pt,160pt) {};
\node[pt, label=below:{\small YY}] ("[1 1 1 1] 0") at (220pt,0pt) {};
\node[pt] ("[1 1 1 1] 1") at (220pt,40pt) {};
\node[pt] ("[1 1 1 1] 2") at (220pt,80pt) {};
\node[pt] ("[1 1 1 1] 3") at (220pt,120pt) {};
\node[pt] ("[1 1 1 1] 4") at (220pt,160pt) {};
\path[->] ("[0 0 0 0] 0") edge[-{Latex[width=3pt]}, line width=1.1pt] ("[0 0 0 0] 1");
\path[->] ("[0 0 0 0] 1") edge[-{Latex[width=3pt]}, line width=1.1pt] ("[0 0 0 0] 2");
\path[->] ("[0 0 0 0] 2") edge[-{Latex[width=3pt]}, line width=1.1pt] ("[0 0 0 0] 3");
\path[->] ("[0 0 0 0] 3") edge[-{Latex[width=3pt]}, line width=1.1pt] ("[0 0 0 0] 4");
\path[->] ("[0 0 0 1] 0") edge[-{Latex[width=3pt]}, line width=1.1pt] ("[0 0 1 1] 1");
\path[->, myblue, dash dot] ("[0 0 0 1] 1") edge[-{Latex[width=3pt]}, line width=1.1pt] ("[0 0 1 1] 2");
\path[->] ("[0 0 0 1] 2") edge[-{Latex[width=3pt]}, line width=1.1pt] ("[0 0 1 1] 3");
\path[->, myblue, dash dot] ("[0 0 0 1] 3") edge[-{Latex[width=3pt]}, line width=1.1pt] ("[0 0 1 1] 4");
\path[->] ("[0 0 1 0] 0") edge[-{Latex[width=3pt]}, line width=1.1pt] ("[0 0 1 0] 1");
\path[->] ("[0 0 1 0] 1") edge[-{Latex[width=3pt]}, line width=1.1pt] ("[0 0 1 0] 2");
\path[->] ("[0 0 1 0] 2") edge[-{Latex[width=3pt]}, line width=1.1pt] ("[0 0 1 0] 3");
\path[->] ("[0 0 1 0] 3") edge[-{Latex[width=3pt]}, line width=1.1pt] ("[0 0 1 0] 4");
\path[->, myblue, dash dot] ("[0 0 1 1] 0") edge[-{Latex[width=3pt]}, line width=1.1pt] ("[0 0 0 1] 1");
\path[->] ("[0 0 1 1] 1") edge[-{Latex[width=3pt]}, line width=1.1pt] ("[0 0 0 1] 2");
\path[->, myblue, dash dot] ("[0 0 1 1] 2") edge[-{Latex[width=3pt]}, line width=1.1pt] ("[0 0 0 1] 3");
\path[->] ("[0 0 1 1] 3") edge[-{Latex[width=3pt]}, line width=1.1pt] ("[0 0 0 1] 4");
\path[->] ("[0 1 0 0] 0") edge[-{Latex[width=3pt]}, line width=1.1pt] ("[0 1 0 0] 1");
\path[->] ("[0 1 0 0] 1") edge[-{Latex[width=3pt]}, line width=1.1pt] ("[0 1 0 0] 2");
\path[->] ("[0 1 0 0] 2") edge[-{Latex[width=3pt]}, line width=1.1pt] ("[0 1 0 0] 3");
\path[->] ("[0 1 0 0] 3") edge[-{Latex[width=3pt]}, line width=1.1pt] ("[0 1 0 0] 4");
\path[->] ("[0 1 0 1] 0") edge[-{Latex[width=3pt]}, line width=1.1pt] ("[0 1 1 1] 1");
\path[->] ("[0 1 0 1] 1") edge[-{Latex[width=3pt]}, line width=1.1pt] ("[0 1 1 1] 2");
\path[->] ("[0 1 0 1] 2") edge[-{Latex[width=3pt]}, line width=1.1pt] ("[0 1 1 1] 3");
\path[->] ("[0 1 0 1] 3") edge[-{Latex[width=3pt]}, line width=1.1pt] ("[0 1 1 1] 4");
\path[->] ("[0 1 1 0] 0") edge[-{Latex[width=3pt]}, line width=1.1pt] ("[0 1 1 0] 1");
\path[->] ("[0 1 1 0] 1") edge[-{Latex[width=3pt]}, line width=1.1pt] ("[0 1 1 0] 2");
\path[->] ("[0 1 1 0] 2") edge[-{Latex[width=3pt]}, line width=1.1pt] ("[0 1 1 0] 3");
\path[->] ("[0 1 1 0] 3") edge[-{Latex[width=3pt]}, line width=1.1pt] ("[0 1 1 0] 4");
\path[->] ("[0 1 1 1] 0") edge[-{Latex[width=3pt]}, line width=1.1pt] ("[0 1 0 1] 1");
\path[->] ("[0 1 1 1] 1") edge[-{Latex[width=3pt]}, line width=1.1pt] ("[0 1 0 1] 2");
\path[->] ("[0 1 1 1] 2") edge[-{Latex[width=3pt]}, line width=1.1pt] ("[0 1 0 1] 3");
\path[->] ("[0 1 1 1] 3") edge[-{Latex[width=3pt]}, line width=1.1pt] ("[0 1 0 1] 4");
\path[->] ("[1 0 0 0] 0") edge[-{Latex[width=3pt]}, line width=1.1pt] ("[1 1 0 0] 1");
\path[->] ("[1 0 0 0] 1") edge[-{Latex[width=3pt]}, line width=1.1pt] ("[1 1 0 0] 2");
\path[->] ("[1 0 0 0] 2") edge[-{Latex[width=3pt]}, line width=1.1pt] ("[1 1 0 0] 3");
\path[->] ("[1 0 0 0] 3") edge[-{Latex[width=3pt]}, line width=1.1pt] ("[1 1 0 0] 4");
\path[->] ("[1 0 0 1] 0") edge[-{Latex[width=3pt]}, line width=1.1pt] ("[1 1 1 1] 1");
\path[->] ("[1 0 0 1] 1") edge[-{Latex[width=3pt]}, line width=1.1pt] ("[1 1 1 1] 2");
\path[->] ("[1 0 0 1] 2") edge[-{Latex[width=3pt]}, line width=1.1pt] ("[1 1 1 1] 3");
\path[->] ("[1 0 0 1] 3") edge[-{Latex[width=3pt]}, line width=1.1pt] ("[1 1 1 1] 4");
\path[->] ("[1 0 1 0] 0") edge[-{Latex[width=3pt]}, line width=1.1pt] ("[1 1 1 0] 1");
\path[->] ("[1 0 1 0] 1") edge[-{Latex[width=3pt]}, line width=1.1pt] ("[1 1 1 0] 2");
\path[->] ("[1 0 1 0] 2") edge[-{Latex[width=3pt]}, line width=1.1pt] ("[1 1 1 0] 3");
\path[->] ("[1 0 1 0] 3") edge[-{Latex[width=3pt]}, line width=1.1pt] ("[1 1 1 0] 4");
\path[->] ("[1 0 1 1] 0") edge[-{Latex[width=3pt]}, line width=1.1pt] ("[1 1 0 1] 1");
\path[->] ("[1 0 1 1] 1") edge[-{Latex[width=3pt]}, line width=1.1pt] ("[1 1 0 1] 2");
\path[->] ("[1 0 1 1] 2") edge[-{Latex[width=3pt]}, line width=1.1pt] ("[1 1 0 1] 3");
\path[->] ("[1 0 1 1] 3") edge[-{Latex[width=3pt]}, line width=1.1pt] ("[1 1 0 1] 4");
\path[->] ("[1 1 0 0] 0") edge[-{Latex[width=3pt]}, line width=1.1pt] ("[1 0 0 0] 1");
\path[->] ("[1 1 0 0] 1") edge[-{Latex[width=3pt]}, line width=1.1pt] ("[1 0 0 0] 2");
\path[->] ("[1 1 0 0] 2") edge[-{Latex[width=3pt]}, line width=1.1pt] ("[1 0 0 0] 3");
\path[->] ("[1 1 0 0] 3") edge[-{Latex[width=3pt]}, line width=1.1pt] ("[1 0 0 0] 4");
\path[->] ("[1 1 0 1] 0") edge[-{Latex[width=3pt]}, line width=1.1pt] ("[1 0 1 1] 1");
\path[->] ("[1 1 0 1] 1") edge[-{Latex[width=3pt]}, line width=1.1pt] ("[1 0 1 1] 2");
\path[->] ("[1 1 0 1] 2") edge[-{Latex[width=3pt]}, line width=1.1pt] ("[1 0 1 1] 3");
\path[->] ("[1 1 0 1] 3") edge[-{Latex[width=3pt]}, line width=1.1pt] ("[1 0 1 1] 4");
\path[->] ("[1 1 1 0] 0") edge[-{Latex[width=3pt]}, line width=1.1pt] ("[1 0 1 0] 1");
\path[->] ("[1 1 1 0] 1") edge[-{Latex[width=3pt]}, line width=1.1pt] ("[1 0 1 0] 2");
\path[->] ("[1 1 1 0] 2") edge[-{Latex[width=3pt]}, line width=1.1pt] ("[1 0 1 0] 3");
\path[->] ("[1 1 1 0] 3") edge[-{Latex[width=3pt]}, line width=1.1pt] ("[1 0 1 0] 4");
\path[->] ("[1 1 1 1] 0") edge[-{Latex[width=3pt]}, line width=1.1pt] ("[1 0 0 1] 1");
\path[->] ("[1 1 1 1] 1") edge[-{Latex[width=3pt]}, line width=1.1pt] ("[1 0 0 1] 2");
\path[->] ("[1 1 1 1] 2") edge[-{Latex[width=3pt]}, line width=1.1pt] ("[1 0 0 1] 3");
\path[->] ("[1 1 1 1] 3") edge[-{Latex[width=3pt]}, line width=1.1pt] ("[1 0 0 1] 4");
\end{tikzpicture}
}

\subfloat[Composition of noisy MCMs($G=\mathrm{CNOT}$, $m=n=1$, with no interleaving single qubit Clifford gates) in the dual space. Multiple transitions are happening in the same time, so fix a start point and an end point, there may be many paths (highlighted by blue dash-dot lines) between them and these paths will superpose together.\label{fig:USIcompose}] {\centering 
\begin{tikzpicture}
[pt/.style={circle,draw,thick,inner sep=0pt,minimum size=0.5em, scale=1}]
\path[->] (-20pt, 0pt) edge[ultra thick, -{Stealth[length=10pt]}] node[sloped, anchor=center, above] {Time} (-20pt, 160pt);
\node[pt, label=below:{\small II}] ("[0 0 0 0] 0") at (0pt,0pt) {};
\node[pt] ("[0 0 0 0] 1") at (0pt,40pt) {};
\node[pt] ("[0 0 0 0] 2") at (0pt,80pt) {};
\node[pt] ("[0 0 0 0] 3") at (0pt,120pt) {};
\node[pt] ("[0 0 0 0] 4") at (0pt,160pt) {};
\node[pt, label=below:{\small IZ}] ("[0 0 0 1] 0") at (40pt,0pt) {};
\node[pt] ("[0 0 0 1] 1") at (40pt,40pt) {};
\node[pt] ("[0 0 0 1] 2") at (40pt,80pt) {};
\node[pt] ("[0 0 0 1] 3") at (40pt,120pt) {};
\node[pt] ("[0 0 0 1] 4") at (40pt,160pt) {};
\node[pt, label=below:{\small ZI}] ("[0 0 1 0] 0") at (60pt,0pt) {};
\node[pt] ("[0 0 1 0] 1") at (60pt,40pt) {};
\node[pt] ("[0 0 1 0] 2") at (60pt,80pt) {};
\node[pt] ("[0 0 1 0] 3") at (60pt,120pt) {};
\node[pt] ("[0 0 1 0] 4") at (60pt,160pt) {};
\node[pt, label=below:{\small ZZ}] ("[0 0 1 1] 0") at (20pt,0pt) {};
\node[pt] ("[0 0 1 1] 1") at (20pt,40pt) {};
\node[pt] ("[0 0 1 1] 2") at (20pt,80pt) {};
\node[pt] ("[0 0 1 1] 3") at (20pt,120pt) {};
\node[pt] ("[0 0 1 1] 4") at (20pt,160pt) {};
\node[pt, label=below:{\small IX}] ("[0 1 0 0] 0") at (80pt,0pt) {};
\node[pt] ("[0 1 0 0] 1") at (80pt,40pt) {};
\node[pt] ("[0 1 0 0] 2") at (80pt,80pt) {};
\node[pt] ("[0 1 0 0] 3") at (80pt,120pt) {};
\node[pt] ("[0 1 0 0] 4") at (80pt,160pt) {};
\node[pt, label=below:{\small IY}] ("[0 1 0 1] 0") at (100pt,0pt) {};
\node[pt] ("[0 1 0 1] 1") at (100pt,40pt) {};
\node[pt] ("[0 1 0 1] 2") at (100pt,80pt) {};
\node[pt] ("[0 1 0 1] 3") at (100pt,120pt) {};
\node[pt] ("[0 1 0 1] 4") at (100pt,160pt) {};
\node[pt, label=below:{\small ZX}] ("[0 1 1 0] 0") at (140pt,0pt) {};
\node[pt] ("[0 1 1 0] 1") at (140pt,40pt) {};
\node[pt] ("[0 1 1 0] 2") at (140pt,80pt) {};
\node[pt] ("[0 1 1 0] 3") at (140pt,120pt) {};
\node[pt] ("[0 1 1 0] 4") at (140pt,160pt) {};
\node[pt, label=below:{\small ZY}] ("[0 1 1 1] 0") at (120pt,0pt) {};
\node[pt] ("[0 1 1 1] 1") at (120pt,40pt) {};
\node[pt] ("[0 1 1 1] 2") at (120pt,80pt) {};
\node[pt] ("[0 1 1 1] 3") at (120pt,120pt) {};
\node[pt] ("[0 1 1 1] 4") at (120pt,160pt) {};
\node[pt, label=below:{\small XI}] ("[1 0 0 0] 0") at (160pt,0pt) {};
\node[pt] ("[1 0 0 0] 1") at (160pt,40pt) {};
\node[pt] ("[1 0 0 0] 2") at (160pt,80pt) {};
\node[pt] ("[1 0 0 0] 3") at (160pt,120pt) {};
\node[pt] ("[1 0 0 0] 4") at (160pt,160pt) {};
\node[pt, label=below:{\small XZ}] ("[1 0 0 1] 0") at (200pt,0pt) {};
\node[pt] ("[1 0 0 1] 1") at (200pt,40pt) {};
\node[pt] ("[1 0 0 1] 2") at (200pt,80pt) {};
\node[pt] ("[1 0 0 1] 3") at (200pt,120pt) {};
\node[pt] ("[1 0 0 1] 4") at (200pt,160pt) {};
\node[pt, label=below:{\small YI}] ("[1 0 1 0] 0") at (240pt,0pt) {};
\node[pt] ("[1 0 1 0] 1") at (240pt,40pt) {};
\node[pt] ("[1 0 1 0] 2") at (240pt,80pt) {};
\node[pt] ("[1 0 1 0] 3") at (240pt,120pt) {};
\node[pt] ("[1 0 1 0] 4") at (240pt,160pt) {};
\node[pt, label=below:{\small YZ}] ("[1 0 1 1] 0") at (280pt,0pt) {};
\node[pt] ("[1 0 1 1] 1") at (280pt,40pt) {};
\node[pt] ("[1 0 1 1] 2") at (280pt,80pt) {};
\node[pt] ("[1 0 1 1] 3") at (280pt,120pt) {};
\node[pt] ("[1 0 1 1] 4") at (280pt,160pt) {};
\node[pt, label=below:{\small XX}] ("[1 1 0 0] 0") at (180pt,0pt) {};
\node[pt] ("[1 1 0 0] 1") at (180pt,40pt) {};
\node[pt] ("[1 1 0 0] 2") at (180pt,80pt) {};
\node[pt] ("[1 1 0 0] 3") at (180pt,120pt) {};
\node[pt] ("[1 1 0 0] 4") at (180pt,160pt) {};
\node[pt, label=below:{\small XY}] ("[1 1 0 1] 0") at (260pt,0pt) {};
\node[pt] ("[1 1 0 1] 1") at (260pt,40pt) {};
\node[pt] ("[1 1 0 1] 2") at (260pt,80pt) {};
\node[pt] ("[1 1 0 1] 3") at (260pt,120pt) {};
\node[pt] ("[1 1 0 1] 4") at (260pt,160pt) {};
\node[pt, label=below:{\small YX}] ("[1 1 1 0] 0") at (300pt,0pt) {};
\node[pt] ("[1 1 1 0] 1") at (300pt,40pt) {};
\node[pt] ("[1 1 1 0] 2") at (300pt,80pt) {};
\node[pt] ("[1 1 1 0] 3") at (300pt,120pt) {};
\node[pt] ("[1 1 1 0] 4") at (300pt,160pt) {};
\node[pt, label=below:{\small YY}] ("[1 1 1 1] 0") at (220pt,0pt) {};
\node[pt] ("[1 1 1 1] 1") at (220pt,40pt) {};
\node[pt] ("[1 1 1 1] 2") at (220pt,80pt) {};
\node[pt] ("[1 1 1 1] 3") at (220pt,120pt) {};
\node[pt] ("[1 1 1 1] 4") at (220pt,160pt) {};
\path[->] ("[0 0 0 0] 0") edge[-{Latex[width=3pt]}, line width=1.1pt] ("[0 0 0 0] 1");
\path[->] ("[0 0 0 0] 0") edge[-{Latex[width=3pt]}, line width=1.1pt] ("[0 0 0 1] 1");
\path[->, myblue, dash dot] ("[0 0 0 0] 1") edge[-{Latex[width=3pt]}, line width=1.1pt] ("[0 0 0 0] 2");
\path[->, myblue, dash dot] ("[0 0 0 0] 1") edge[-{Latex[width=3pt]}, line width=1.1pt] ("[0 0 0 1] 2");
\path[->] ("[0 0 0 0] 2") edge[-{Latex[width=3pt]}, line width=1.1pt] ("[0 0 0 0] 3");
\path[->, myblue, dash dot] ("[0 0 0 0] 2") edge[-{Latex[width=3pt]}, line width=1.1pt] ("[0 0 0 1] 3");
\path[->] ("[0 0 0 0] 3") edge[-{Latex[width=3pt]}, line width=1.1pt] ("[0 0 0 0] 4");
\path[->] ("[0 0 0 0] 3") edge[-{Latex[width=3pt]}, line width=1.1pt] ("[0 0 0 1] 4");
\path[->, myblue, dash dot] ("[0 0 1 1] 0") edge[-{Latex[width=3pt]}, line width=1.1pt] ("[0 0 0 0] 1");
\path[->, myblue, dash dot] ("[0 0 1 1] 0") edge[-{Latex[width=3pt]}, line width=1.1pt] ("[0 0 0 1] 1");
\path[->] ("[0 0 1 1] 1") edge[-{Latex[width=3pt]}, line width=1.1pt] ("[0 0 0 0] 2");
\path[->] ("[0 0 1 1] 1") edge[-{Latex[width=3pt]}, line width=1.1pt] ("[0 0 0 1] 2");
\path[->] ("[0 0 1 1] 2") edge[-{Latex[width=3pt]}, line width=1.1pt] ("[0 0 0 0] 3");
\path[->, myblue, dash dot] ("[0 0 1 1] 2") edge[-{Latex[width=3pt]}, line width=1.1pt] ("[0 0 0 1] 3");
\path[->] ("[0 0 1 1] 3") edge[-{Latex[width=3pt]}, line width=1.1pt] ("[0 0 0 0] 4");
\path[->] ("[0 0 1 1] 3") edge[-{Latex[width=3pt]}, line width=1.1pt] ("[0 0 0 1] 4");
\path[->] ("[0 0 1 0] 0") edge[-{Latex[width=3pt]}, line width=1.1pt] ("[0 0 1 0] 1");
\path[->] ("[0 0 1 0] 0") edge[-{Latex[width=3pt]}, line width=1.1pt] ("[0 0 1 1] 1");
\path[->] ("[0 0 1 0] 1") edge[-{Latex[width=3pt]}, line width=1.1pt] ("[0 0 1 0] 2");
\path[->] ("[0 0 1 0] 1") edge[-{Latex[width=3pt]}, line width=1.1pt] ("[0 0 1 1] 2");
\path[->, myblue, dash dot] ("[0 0 1 0] 2") edge[-{Latex[width=3pt]}, line width=1.1pt] ("[0 0 1 0] 3");
\path[->] ("[0 0 1 0] 2") edge[-{Latex[width=3pt]}, line width=1.1pt] ("[0 0 1 1] 3");
\path[->] ("[0 0 1 0] 3") edge[-{Latex[width=3pt]}, line width=1.1pt] ("[0 0 1 0] 4");
\path[->, myblue, dash dot] ("[0 0 1 0] 3") edge[-{Latex[width=3pt]}, line width=1.1pt] ("[0 0 1 1] 4");
\path[->] ("[0 0 0 1] 0") edge[-{Latex[width=3pt]}, line width=1.1pt] ("[0 0 1 0] 1");
\path[->] ("[0 0 0 1] 0") edge[-{Latex[width=3pt]}, line width=1.1pt] ("[0 0 1 1] 1");
\path[->, myblue, dash dot] ("[0 0 0 1] 1") edge[-{Latex[width=3pt]}, line width=1.1pt] ("[0 0 1 0] 2");
\path[->, myblue, dash dot] ("[0 0 0 1] 1") edge[-{Latex[width=3pt]}, line width=1.1pt] ("[0 0 1 1] 2");
\path[->, myblue, dash dot] ("[0 0 0 1] 2") edge[-{Latex[width=3pt]}, line width=1.1pt] ("[0 0 1 0] 3");
\path[->] ("[0 0 0 1] 2") edge[-{Latex[width=3pt]}, line width=1.1pt] ("[0 0 1 1] 3");
\path[->] ("[0 0 0 1] 3") edge[-{Latex[width=3pt]}, line width=1.1pt] ("[0 0 1 0] 4");
\path[->, myblue, dash dot] ("[0 0 0 1] 3") edge[-{Latex[width=3pt]}, line width=1.1pt] ("[0 0 1 1] 4");
\path[->] ("[1 1 0 0] 0") edge[-{Latex[width=3pt]}, line width=1.1pt] ("[1 0 0 0] 1");
\path[->] ("[1 1 0 0] 0") edge[-{Latex[width=3pt]}, line width=1.1pt] ("[1 0 0 1] 1");
\path[->] ("[1 1 0 0] 1") edge[-{Latex[width=3pt]}, line width=1.1pt] ("[1 0 0 0] 2");
\path[->] ("[1 1 0 0] 1") edge[-{Latex[width=3pt]}, line width=1.1pt] ("[1 0 0 1] 2");
\path[->] ("[1 1 0 0] 2") edge[-{Latex[width=3pt]}, line width=1.1pt] ("[1 0 0 0] 3");
\path[->] ("[1 1 0 0] 2") edge[-{Latex[width=3pt]}, line width=1.1pt] ("[1 0 0 1] 3");
\path[->] ("[1 1 0 0] 3") edge[-{Latex[width=3pt]}, line width=1.1pt] ("[1 0 0 0] 4");
\path[->] ("[1 1 0 0] 3") edge[-{Latex[width=3pt]}, line width=1.1pt] ("[1 0 0 1] 4");
\path[->] ("[1 1 1 1] 0") edge[-{Latex[width=3pt]}, line width=1.1pt] ("[1 0 0 0] 1");
\path[->] ("[1 1 1 1] 0") edge[-{Latex[width=3pt]}, line width=1.1pt] ("[1 0 0 1] 1");
\path[->] ("[1 1 1 1] 1") edge[-{Latex[width=3pt]}, line width=1.1pt] ("[1 0 0 0] 2");
\path[->] ("[1 1 1 1] 1") edge[-{Latex[width=3pt]}, line width=1.1pt] ("[1 0 0 1] 2");
\path[->] ("[1 1 1 1] 2") edge[-{Latex[width=3pt]}, line width=1.1pt] ("[1 0 0 0] 3");
\path[->] ("[1 1 1 1] 2") edge[-{Latex[width=3pt]}, line width=1.1pt] ("[1 0 0 1] 3");
\path[->] ("[1 1 1 1] 3") edge[-{Latex[width=3pt]}, line width=1.1pt] ("[1 0 0 0] 4");
\path[->] ("[1 1 1 1] 3") edge[-{Latex[width=3pt]}, line width=1.1pt] ("[1 0 0 1] 4");
\path[->] ("[1 1 1 0] 0") edge[-{Latex[width=3pt]}, line width=1.1pt] ("[1 0 1 0] 1");
\path[->] ("[1 1 1 0] 0") edge[-{Latex[width=3pt]}, line width=1.1pt] ("[1 0 1 1] 1");
\path[->] ("[1 1 1 0] 1") edge[-{Latex[width=3pt]}, line width=1.1pt] ("[1 0 1 0] 2");
\path[->] ("[1 1 1 0] 1") edge[-{Latex[width=3pt]}, line width=1.1pt] ("[1 0 1 1] 2");
\path[->] ("[1 1 1 0] 2") edge[-{Latex[width=3pt]}, line width=1.1pt] ("[1 0 1 0] 3");
\path[->] ("[1 1 1 0] 2") edge[-{Latex[width=3pt]}, line width=1.1pt] ("[1 0 1 1] 3");
\path[->] ("[1 1 1 0] 3") edge[-{Latex[width=3pt]}, line width=1.1pt] ("[1 0 1 0] 4");
\path[->] ("[1 1 1 0] 3") edge[-{Latex[width=3pt]}, line width=1.1pt] ("[1 0 1 1] 4");
\path[->] ("[1 1 0 1] 0") edge[-{Latex[width=3pt]}, line width=1.1pt] ("[1 0 1 0] 1");
\path[->] ("[1 1 0 1] 0") edge[-{Latex[width=3pt]}, line width=1.1pt] ("[1 0 1 1] 1");
\path[->] ("[1 1 0 1] 1") edge[-{Latex[width=3pt]}, line width=1.1pt] ("[1 0 1 0] 2");
\path[->] ("[1 1 0 1] 1") edge[-{Latex[width=3pt]}, line width=1.1pt] ("[1 0 1 1] 2");
\path[->] ("[1 1 0 1] 2") edge[-{Latex[width=3pt]}, line width=1.1pt] ("[1 0 1 0] 3");
\path[->] ("[1 1 0 1] 2") edge[-{Latex[width=3pt]}, line width=1.1pt] ("[1 0 1 1] 3");
\path[->] ("[1 1 0 1] 3") edge[-{Latex[width=3pt]}, line width=1.1pt] ("[1 0 1 0] 4");
\path[->] ("[1 1 0 1] 3") edge[-{Latex[width=3pt]}, line width=1.1pt] ("[1 0 1 1] 4");
\end{tikzpicture}
}
\caption{Comparison between the composition of noisy Clifford gates ($\widetilde{\mathcal{G}}$ in \cref{eq:PCd}) and noisy MCMs ($\mathcal{T}$ in \cref{eq:USId}) in the dual space.}
\label{fig:virtSpaceComp}
\end{figure}